\documentclass[a4paper]{article}
\usepackage[utf8]{inputenc}
\usepackage{latexsym, amssymb, dsfont}
\usepackage{amsmath,amsthm}
\usepackage{fullpage}
\usepackage{enumerate}
\usepackage{tikz}
\usepackage{ifthen}

\newboolean{PGFPlots}
\setboolean{PGFPlots}{false}

\usetikzlibrary{circuits.logic.US}
\pgfdeclarelayer{background}
\pgfsetlayers{background,main}
\usetikzlibrary{arrows,shapes,decorations,decorations.pathreplacing,decorations.pathmorphing,intersections,calc,fit}

\ifthenelse{\boolean{PGFPlots}}{
\usepackage{pgfplots}
\usepgfplotslibrary{units}
\usetikzlibrary{external}
\tikzexternalize
\tikzset{external/only named=true}
}{}

\newtheorem{theorem}{Theorem}
\newtheorem{lemma}[theorem]{Lemma}
\theoremstyle{definition}

\newtheorem{definition}[theorem]{Definition}

\gdef\dash---{\thinspace---\hskip.16667em\relax} 
\gdef\op|{\,|\;}

\newcommand{\figref}[1]{Fig.~\ref{#1}}

\newcommand{\bfno}[1]{\noindent{\bf #1}}

\newcommand{\IN}{\mathds{N}}

\newcommand{\IR}{\mathds{R}}

\newcommand{\dup}{\delta_\uparrow}
\newcommand{\ddo}{\delta_\downarrow}
\newcommand{\fup}{f_\uparrow}
\newcommand{\fdo}{f_\downarrow}

\newboolean{SHORTversion}
\setboolean{SHORTversion}{false}

\usetikzlibrary{petri}
\tikzstyle{binary place}=[place,circle, double]
\tikzstyle{node}=[circle,draw=black,thick,minimum size=9mm]
\tikzstyle{dest}=[circle,draw=black!50,fill=black!20,thick,minimum
  size=9mm]
\tikzstyle{post}=[->,thick]
\tikzstyle{pre}=[<-,thick]
\tikzstyle{every transition}=[fill,minimum width=1cm,minimum height=2mm]
\tikzstyle{Atransition}=[transition,fill,minimum width=1cm,minimum height=2mm]
\tikzstyle{Otransition}=[transition,fill=white,minimum width=1cm,minimum height=2mm]
\tikzstyle{THtransition}=[transition,fill=white,minimum width=4mm,minimum height=1cm]

\tikzstyle{Tdelay} = [draw, rectangle, rounded corners,
    minimum height=4mm, minimum width=20mm]

\tikzstyle{Tfunction} = [draw, rectangle,
    minimum height=4mm, minimum width=4mm]

\tikzstyle{Tsignal} = [draw,fill=black,circle, size=1mm]

\tikzstyle{ra} = [draw,thick,double,double distance=1.0pt,->]
\tikzstyle{r} = [draw,->,line width=0.5pt]

\def\figanvth{0.4}
\def\figaninstart{0.1}
\def\figaninlen{0.28}
\def\figandelay{0.05}
\def\figanfup[#1]{(1-exp(-(#1)/0.15))*(1-exp(-(#1)/0.16))}
\def\figanfdown[#1]{(1-(1-exp(-((#1)^2)/0.01))*(1-exp(-((#1)^2)/0.1)))}
\def\figancaptpos{0.9}

\def\tikzfigureanalogy{
\begin{tikzpicture}[
xscale=5.4,
yscale=2.0,
>=latex'
]

\coordinate (origin) at (0,0);

\path [draw,->]
(0,0) -- (0,1.1) node (xaxis) [left] {\small $u(t)$};
\path [draw,->]
(0,0) -- (1.05,0)   node (yaxis) [above] {\small $t$};

\path [draw, dashdotted, very thin, name path=vth]
(0,\figanvth) node [left] {\small $V_{th}$} -- (1,\figanvth);

\path [draw, densely dotted, line width=0.3pt]
(\figaninstart,0) -- ++(0,1) node [above, xshift=-1pt, yshift=-1pt] {\small $u_i$} --
++(\figaninlen,0) -- ++(0,-1) node (dt3) [coordinate] {};

\path [draw, densely dashed, very thin]
({\figaninstart+\figandelay},0) -- ++(0,1) node [above, xshift=3pt, yshift=-1pt] {\small $u_d$} --
++(\figaninlen,0) -- ++(0,-1);

\path [draw, thick, xshift={(\figaninstart+\figandelay)*1cm}, name path=fup]
plot[smooth, domain=0:\figaninlen] (\x,{\figanfup[\x]});
\path [draw, thick, dotted, xshift={(\figaninstart+\figandelay)*1cm}, name path=fupdot]
plot[smooth, domain=\figaninlen:{1-\figaninstart-\figandelay}] (\x,{\figanfup[\x]});
\path [name path=func]
plot[smooth, domain=0:1] (\x,{\figanfdown[\x]});
\path [name path=const]
(0,{\figanfup[\figaninlen]}) -- +(1,0);
\path [name intersections={of=func and const}]
(intersection-1) node (finversept) [coordinate] {};
\newdimen\finverse
\pgfextractx{\finverse}{\pgfpointanchor{finversept}{center}}
\pgfmathsetlengthmacro{\figanfdownshift}{(\figaninstart+\figandelay+\figaninlen)*1cm-\finverse}
\path [draw, thick, dotted, xshift=\figanfdownshift]
plot[smooth, domain=0:{\finverse/1cm}] (\x,{\figanfdown[\x]});
\path [draw, thick, xshift=\figanfdownshift, name path=fdown]
plot[smooth, domain={\finverse/1cm}:{(1cm-\figanfdownshift)/1cm}] (\x,{\figanfdown[\x]});

\path [draw, thin, name intersections={of=fup and vth, by={up}}, name intersections={of=fdown and vth, by={down}}]
(origin -| up) node (dt2) [coordinate] {} -- ++(0,1) node [above, xshift=3pt, yshift=-1pt] {\small $u_o$} -| 
(origin -| down) node (dt4) [coordinate] {};

\path [name path=fcapts]
(\figancaptpos,0) -- ++(0,1);
\path [name intersections={of=fcapts and fdown, by={fd}}, name intersections={of=fcapts and fupdot, by={fu}}]
(fd) node [above] {\small $f_\downarrow$}
(fu) node [below] {\small $f_\uparrow$};

\foreach \dt/\name in {{(\figaninstart,0)/1}, {(dt2)/2}, {(dt3)/3}, {(dt4)/4}} {
\path [draw, very thin, shorten <= 1pt, shorten >=-2pt]
\dt -- ++(0,-0.1) node (smth-\name) [coordinate] {};
}

\path [draw, <-, very thin]
(smth-1) -- node [anchor=north] {\tiny $T_1$} (smth-1 -| origin);
\path [draw, ->, very thin, shorten <=0.6pt]
(smth-1) -- node [anchor=north] {\tiny $\delta_\uparrow(T_1)$} (smth-2);
\path [draw, ->, very thin, shorten <=0.6pt]
(smth-2) -- node [anchor=north] {\tiny $T_2$} (smth-3);
\path [draw, ->, very thin, shorten <=0.6pt]
(smth-3) -- node [anchor=north] {\tiny $\delta_\downarrow(T_2)$} (smth-4);
\end{tikzpicture}
}

\def\figcirampin{38}
\def\figcirampinlen{0.25}
\def\figcirampvoltlen{0.27}
\def\figcirampcap{0.3}
\def\figcirdista{0.45}
\def\figcirdistb{0.55}
\def\figcirdistc{1}
\def\figciroutlen{0.25}

\def\tikzfigurecircuit{
\begin{tikzpicture}[
cmp/.style={draw, isosceles triangle, isosceles triangle stretches, minimum width=1.1cm, minimum height=1.0cm},
xscale=1.0,
yscale=1.0,
>=latex'
]

\path
(0,0) node (cmp1) [cmp] {};
\path [draw]
(cmp1.{180-\figcirampin}) node [anchor=west] {\footnotesize $+$} --
++(-\figcirampinlen,0) node (cmp1-west) [coordinate] {};
\path [draw]
(cmp1.{180+\figcirampin}) node [anchor=west] {\footnotesize $-$} -- ++(-\figcirampinlen,0);

\path [draw, shorten <=-0.5pt]
(cmp1.east) -- node [anchor=south] {\small $u_i$}
++(\figcirdista,0) node (del) [anchor=west, draw, minimum height=0.4cm, minimum width=0.8cm,
                      inner sep=1pt, rounded corners=0.2cm] {\small $T_p$};

\path [draw]
(del.east) -- node [anchor=south] {\small $u_d$}
++(\figcirdistb,0) node (srl) [draw, minimum size=0.8cm, anchor=west] {};
\path [draw, scale=0.45]
(srl) ++(-0.5,-0.5) .. controls +(0.3,0) and +(-0.8,0) .. +(1,1);

\path [draw]
(srl.east) -- node [anchor=south] {\small $u_r$}
++(\figcirdistc,0) node (cmp2) [cmp, anchor={180-\figcirampin}] {}
                     node [anchor=west] {\footnotesize $+$};
\path [draw]
(cmp2.{180+\figcirampin}) node [anchor=west] {\footnotesize $-$} --
++(-\figcirampinlen,0) node [anchor=east, inner sep=1pt] {\small $V_{th}$};
\path [draw, shorten <=-0.5pt]
(cmp2.east) --
++(\figciroutlen,0) node [anchor=290] {\small $u_o$};

\foreach \cmpx in {cmp1, cmp2} {
\path [draw]
(\cmpx.50) -- ++(0,\figcirampvoltlen) node (\cmpx-1v) [anchor=south] {\small $1V$};
\path [draw]
(\cmpx.310) -- ++(0,-\figcirampvoltlen) node (\cmpx-0v) [anchor=north] {\small $0V$};
\path
($(\cmpx)+(\figcirampcap,0)$) node {\small $\infty$};
}

\node at ($ (cmp1.east)-(0.3,0) $) (seast) {};

\node[fill=white, fill opacity=0.6, inner sep=0, fit=(cmp1-1v) (cmp1-0v) (cmp1-west) (seast)] {};
\end{tikzpicture}
}

\def\figurangle{320}
\def\figurlen{1.4}
\def\figurdist{4.5}

\def\tikzfigureunrolling{
\begin{tikzpicture}[
gate/.style={draw, fill, circle, minimum size=3pt, inner sep=0},
xscale=1.5,
yscale=0.9,
>=latex'
]

\path
(0,0) node (i) [gate] {}
+(0:\figurlen) node (x) [gate] {}
+(0:{\figurlen/2}) node (ix) [gate] {}
+(\figurangle:{\figurlen/2}) node (iy) [gate] {}
++(\figurangle:\figurlen) node (y) [gate] {}
+(0:{\figurlen/2}) node (yo) [gate] {}
++(0:\figurlen) node (o) [gate] {}
+(\figurangle:{-\figurlen/2}) node (xo) [gate] {};

\foreach \from/\to in {i/ix, ix/x, x/xo, xo/o, i/iy, iy/y, y/yo, yo/o} {
\path [draw, ->] (\from) -- (\to);
}
\path [draw, ->, scale=0.8]
(y) .. controls +(-1,-1) and +(1,-1) .. (y)
node (yy) [midway,gate] {};

\foreach \what/\anchor in {i/south, x/south, y/south, o/north} {
\node [anchor=\anchor] at (\what) {\footnotesize $\what$};
}

\pgftransformshift{\pgfpoint{+1.0cm}{0cm}} 

\path
(\figurdist,0) node (i1) [gate] {}
+(0:{\figurlen/2}) node (i1x1) [gate] {}
++(0:\figurlen) node (x1) [gate] {}
+(\figurangle:{\figurlen/2}) node (x1o1) [gate] {}
++(\figurangle:\figurlen) node (o1) [gate] {}
+(0:{-\figurlen/2}) node (y1o1) [gate] {}
++(0:-\figurlen) node (y1) [gate] {}
+(0:{-\figurlen/2}) node (y2y1) [gate] {}
++(0:-\figurlen) node (y2) [gate] {}
+(0:{-\figurlen/2}) node (ot1y2) [gate] {}
+(0:-\figurlen) node (ot1) [gate] {}
+(\figurangle:{-\figurlen/2}) node (ot2y2) [gate] {}
+(\figurangle:-\figurlen) node (ot2) [gate] {};

\path
(i1)
+(\figurangle:{\figurlen/2}) node (i1y1) [gate] {};

\foreach \from/\to in {ot2/ot2y2, ot2y2/y2, ot1/ot1y2, ot1y2/y2, y2/y2y1, y2y1/y1,
                       y1/y1o1, y1o1/o1, i1/i1y1, i1y1/y1, i1/i1x1, i1x1/x1, x1/x1o1,
                       x1o1/o1} {
\path [draw, ->] (\from) -- (\to);
}

\foreach \where/\anchor/\text in
{ot2/south/$\tilde{0}^{(2)}$, ot1/north/$\tilde{0}^{(1)}$,
y2/north/$y^{(2)}$, i1/south/$i$, x1/south/$x^{(1)}$,
y1/north/$y^{(1)}$, o1/north/$o^{(1)}$} {
\node [anchor=\anchor] at (\where) {\footnotesize \text};
}
\end{tikzpicture}
}

\def\figchtimedist{0.35}
\def\figsigheight{0.7}
\def\figsigoffset{0.1}
\def\figchdist{1.5}

\def\tikzfigurechannelintro{
\begin{tikzpicture}[>=latex',scale=1.0,every node/.style={transform shape}]

\coordinate (origin1) at (0,0);
\coordinate (start1) at (0,\figsigoffset);
\coordinate (origin2) at (0,-\figchdist);
\coordinate (start2) at (0,-\figchdist+\figsigoffset+\figsigheight);
\coordinate (timeorigin) at (0,{-\figchdist-\figchtimedist});
\foreach \plot/\xcapt in {1/$\text{in}(t)$, 2/$\text{out}(t)$} {
\path [draw,->]
(origin\plot) -- +(0,1.1) node (xaxis) [left,yshift=-5pt] {\small \xcapt};
\path [draw,->]
(origin\plot) -- +(8.00,0) node (yaxis) [above] {\small $t$};
}
\draw[very thick] (start1) -- ++(4,0) -- ++(0,\figsigheight) -- ++(4,0);
\draw[very thick] (start2) -- ++(3,0) -- ++(0,-\figsigheight) -- ++(3,0) --
++(0,\figsigheight) -- ++(2,0);
\draw[densely dashed] ($(origin1)+(4,0)$) -- ($(origin2)+(4,-3*\figsigoffset)$);
\draw[densely dashed] ($(origin2)+(3,0)$) -- ++(0,-3*\figsigoffset);
\draw[densely dashed] ($(origin2)+(6,0)$) -- ++(0,-3*\figsigoffset);
\node (T) at ($(origin2) +(3.5,-2*\figsigoffset) $) {$T$};
\node (dT) at ($(origin2) +(5,-2*\figsigoffset) $) {$\delta(T)$};
\draw[->] (T) -- ($(origin2)+(4,-2*\figsigoffset)$);
\draw[->] (T) -- ($(origin2)+(3,-2*\figsigoffset)$);
\draw[->] (dT) -- ($(origin2)+(6,-2*\figsigoffset)$);
\draw[->] (dT) -- ($(origin2)+(4,-2*\figsigoffset)$);
\draw[->,densely dashed,shorten >=2pt,shorten <=2pt] ($(start1) + (4,\figsigheight)$) -- ($(start2) +
(6,0)$);
\draw[->,densely dashed,shorten >=2pt,shorten <=2pt] ($(start1) + (0,0.5*\figsigheight)$) -- ($(start2) +
(3,0)$);

\end{tikzpicture}
}

\def\tikzfigurechannelsection{
\begin{tikzpicture}[>=latex',scale=1.0,every node/.style={transform shape}]

\coordinate (origin1) at (0,0);
\coordinate (start1) at (0,\figsigoffset+\figsigheight);
\coordinate (origin2) at (0,-\figchdist);
\coordinate (start2) at (0,-\figchdist+\figsigoffset+\figsigheight);
\coordinate (timeorigin) at (0,{-\figchdist-\figchtimedist});
\foreach \plot/\xcapt in {1/$\text{in}(t)$, 2/$\text{out}(t)$} {
\path [draw,->]
(origin\plot) -- +(0,1.1) node (xaxis) [left,yshift=-5pt] {\small \xcapt};
\path [draw,->]
(origin\plot) -- +(8.00,0) node (yaxis) [above] {\small $t$};
}
\draw[very thick] (start1) -- ++(1.5,0) -- ++(0,-\figsigheight) -- ++(2,0) --
++(0,\figsigheight) -- ++(4.5,0);
\draw[very thick] (start2) -- ++(5,0) -- ++(0,-\figsigheight) -- ++(1,0) --
++(0,\figsigheight) -- ++(2,0);
\draw[densely dashed] ($(origin1)+(3.5,0)$) -- ($(origin2)+(3.5,-7*\figsigoffset)$);
\draw[densely dashed] ($(origin2)+(5,0)$) -- ++(0,-3*\figsigoffset);
\draw[densely dashed] ($(origin2)+(6,0)$) -- ++(0,-7*\figsigoffset);
\node (T) at ($(origin2) +(4.25,-2*\figsigoffset) $) {$(-T)$};
\node (dT) at ($(origin2) +(4.75,-6*\figsigoffset) $) {$\delta(T)$};
\draw[->] (T) -- ($(origin2)+(5,-2*\figsigoffset)$);
\draw[->] (T) -- ($(origin2)+(3.5,-2*\figsigoffset)$);
\draw[->] (dT) -- ($(origin2)+(6,-6*\figsigoffset)$);
\draw[->] (dT) -- ($(origin2)+(3.5,-6*\figsigoffset)$);
\draw[->,densely dashed,shorten >=2pt,shorten <=2pt] ($(start1) + (3.5,0)$) -- ($(start2) +
(6,0)$);
\draw[->,densely dashed,shorten >=2pt,shorten <=2pt] ($(start1) + (1.5,0)$) -- ($(start2) +
(5,0)$);

\end{tikzpicture}
}

\def\figgawidth{2.8}
\def\figgaheight{1.7}
\def\figgayshift{0.15}
\def\figgaboolwidth{0.75}
\def\figgaboolheight{0.9}

\def\figgaconstwidth{0.5}
\def\figgaconstheight{0.5}

\def\figgamuxwidth{0.45}
\def\figgamuxheight{1}
\def\figgamuxindist{50}
\def\figgamuxtoout{0.25}
\def\figgaoutlen{0.25}
\def\figgarstlen{0.15}
\def\figgainlen{0.25}

\def\tikzfiguregate{
\begin{tikzpicture}[
gate/.style={draw, inner sep=0},
xscale=1,
yscale=1,
>=latex'
]

\coordinate (origin) at (0,0);
\coordinate (topright) at (\figgawidth,\figgaheight);

\path [draw] (origin) rectangle (topright);

\path [draw]
($(origin -| topright)!0.5!(topright) + (0,\figgayshift)$)
+(\figgaoutlen,0) node [anchor=west] {\small $v$} --
+(-\figgamuxtoout,0) node (mux) [draw, trapezium, anchor=top side,
shape border rotate=270, minimum width=\figgamuxheight*1cm, 
minimum height=\figgamuxwidth*1cm, trapezium stretches body] {};

\foreach \muxin/\muxincpt/\name/\cpt in {1/0/bool/$b$, -1/1/const/$I$} {
\path [draw] 
(mux.{180+\figgamuxindist*\muxin}) node [anchor=west, xshift=-1.5pt] {\footnotesize \muxincpt} --
++(-\csname figga\name tomux\endcsname, 0)
node (\name) [gate, anchor=east, minimum width=\csname figga\name width\endcsname*1cm,
minimum height=\csname figga\name height\endcsname*1cm] {\small \cpt};
}

\path [draw, densely dashed, thin]
(mux.north) -- (mux.north |- topright) --
++(0, \figgarstlen) node [anchor=south, inner sep=1.5pt] {\small\em reset};

\foreach \pos in {0.125, 0.375, 0.625, 0.875} {
\path [draw]
($(bool.north west)!\pos!(bool.south west)$) coordinate (tmp) --
(tmp -| origin) -- ++(-\figgainlen,0);
}
\end{tikzpicture}
}

\def\tikzfigurechannelbranch{
\begin{tikzpicture}[
gate/.style={draw, fill, circle, minimum size=3pt, inner sep=0},
gate2/.style={draw, inner sep=0},
xscale=1.5,
yscale=1.2,
>=latex'
]

\coordinate (origin) at (0,0);
\coordinate (topright) at (\figgawidth,\figgaheight);

\path [draw]
($(origin -| topright)!0.5!(topright) + (0,\figgayshift)$)
+(\figgaoutlen,0) --
+(-\figgamuxtoout,0) node (mux) [draw, trapezium, anchor=top side,
shape border rotate=270, minimum width=\figgamuxheight*1cm, 
minimum height=\figgamuxwidth*1cm, trapezium stretches body] {};

\path[draw]
($(origin -| topright)!0.5!(topright) + (0,\figgayshift)$)
+(\figgaoutlen,0)
node (c1) [draw, rectangle, rounded corners, anchor=west, minimum width=\figgamuxheight*1cm, 
minimum height=\figgamuxwidth*1cm] {$c_1$};

\path[draw]
(c1) --
++(1,0)
node (outc1) [gate2, anchor=west, minimum width=\figgaboolwidth*1cm,
minimum height=\figgaboolheight*1cm] {\small $b$};

\path[draw]
(outc1)
-- ++(0.8,0)
node[anchor=west] (y) {\small $w$};

\foreach \muxin/\muxincpt/\name/\cpt in {1/0/bool/$b$, -1/1/const/$I$} {
\path [draw] 
(mux.{180+\figgamuxindist*\muxin}) node [anchor=west, xshift=-1.5pt] {\footnotesize \muxincpt} --
++(-\csname figga\name tomux\endcsname, 0)
node (\name) {}; 
}

\path [draw]
(const)
node (huhu) [gate2, anchor=east, minimum width=\figgaconstwidth*1cm,
minimum height=\figgaconstheight*1cm] {\small $I$};

\path [draw, densely dashed, thin]
(mux.south) -- (mux.south |- origin)
++(0, -\figgarstlen+0.1) node [anchor=south, inner sep=1.5pt] {};

\foreach \pos in {0.23, 0.76} {
\path [draw]
($(outc1.north west)!\pos!(outc1.south west)$) coordinate (tmp) --
++(-\figgainlen,0);
}

\pgftransformshift{\pgfpoint{0cm}{-2cm}} 

\coordinate (origin) at (0,0);
\coordinate (topright) at (\figgawidth,\figgaheight);

\path [draw]
($(origin -| topright)!0.5!(topright) + (0,\figgayshift)$)
+(\figgaoutlen,0) --
+(-\figgamuxtoout,0) node (mux) [draw, trapezium, anchor=top side,
shape border rotate=270, minimum width=\figgamuxheight*1cm, 
minimum height=\figgamuxwidth*1cm, trapezium stretches body] {};

\path[draw]
($(origin -| topright)!0.5!(topright) + (0,\figgayshift)$)
+(\figgaoutlen,0)
node (c1) [draw, rectangle, rounded corners, anchor=west, minimum width=\figgamuxheight*1cm, 
minimum height=\figgamuxwidth*1cm] {$c_2$};

\path[draw]
(c1) --
++(1,0)
node (outc1) [gate2, anchor=west, minimum width=\figgaboolwidth*1cm,
minimum height=\figgaboolheight*1cm] {\small $b'$};

\path[draw]
(outc1)
-- ++(0.8,0)
node[anchor=west] (y) {\small $z$};

\foreach \muxin/\muxincpt/\name/\cpt in {1/0/bool/$b$, -1/1/const/$I$} {
\path [draw] 
(mux.{180+\figgamuxindist*\muxin}) node [anchor=west, xshift=-1.5pt] {\footnotesize \muxincpt} --
++(-\csname figga\name tomux\endcsname, 0)
node (\name) {}; 
}

\path [draw]
(const)
node (huhu) [gate2, anchor=east, minimum width=\figgaconstwidth*1cm,
minimum height=\figgaconstheight*1cm] {\small $I$};

\path [draw]
(bool.center) --
++(0,2cm) node [circle,inner sep=0.7pt,fill=black,draw] {} --
++(-0.2,0) node [anchor=east, xshift=2pt] {\small $v$};

\path [draw, densely dashed, thin]
(mux.north) -- (mux.north |- topright) --
++(0, \figgarstlen) node [anchor=south, inner sep=1.5pt] {\small\em reset};

\foreach \pos in {0.23, 0.76} {
\path [draw]
($(outc1.north west)!\pos!(outc1.south west)$) coordinate (tmp) --
++(-\figgainlen,0);
}

\pgftransformshift{\pgfpoint{-1.2cm}{2cm}} 

\path
(0,0) node (x) [gate] {} node [below] {$v$}
+(0.5,0.5) node (c1) [gate] {} node [anchor=north, yshift=20pt, text height=\baselineskip] {$c_1$}
+(1,0.5) node (y) [gate] {} node [anchor=north, yshift=20pt, text height=\baselineskip] {$w$}
+(0.5,-0.5) node (c2) [gate] {} node [below=-7pt, text height=\baselineskip] {$c_2$}
+(1,-0.5) node (z) [gate] {} node [below=-7pt, text height=\baselineskip] {$z$};

\foreach \from/\to in {x/c1, c1/y, x/c2, c2/z} {
\path [draw, ->] (\from) -- (\to);
}

\draw (1.6,0)
node {$\equiv$};

\end{tikzpicture}
}

\title{Faithful Glitch Propagation in Binary Circuit Models}

\author{Matthias Függer\textsuperscript{1} \and Robert Najvirt\textsuperscript{1} \and Thomas
Nowak\textsuperscript{2} \and Ulrich
Schmid\textsuperscript{1}}
\date{\textsuperscript{1} ECS Group, TU Wien, Austria\\
\textsuperscript{2} \'Ecole normale sup\'erieure,
Paris, France}

\begin{document}
\maketitle
\begin{abstract}
Modern digital circuit design relies on fast digital timing simulation
     tools and, hence, on accurate binary-valued circuit models that
     faithfully model signal propagation, even throughout a complex
     design. Of particular importance is the ability to trace glitches
     and other short pulses, as their presence/absence may even affect
     a circuit's correctness.
Unfortunately, it was recently proved [F\"ugger et al., ASYNC'13] that no 
existing binary-valued circuit model
     proposed so far, including the two most commonly used pure and
     inertial delay channels, faithfully captures glitch propagation: For the
     simple Short-Pulse Filtration (SPF) problem, which is related to
     a circuit's ability to suppress a single glitch, we showed that
     the quite broad class of bounded single-history channels either
     contradict the unsolvability of SPF in bounded time or the
     solvability of SPF in unbounded time in physical circuits.

In this paper, we propose a class of binary circuit models that do not
     suffer from this deficiency: Like bounded single-history
     channels, our \emph{involution channels\/} involve delays that may
     depend on the time of the previous output transition. Their
characteristic property are delay functions which are based on
     involutions, i.e., functions that form their own inverse.
A concrete example of such a delay function, which 
     is derived from a generalized first-order analog circuit model,
     reveals that this is not an unrealistic assumption.
We prove that, in sharp contrast to what is possible with 
bounded single-history channels, 
SPF cannot be solved in bounded time due to the nonexistence of a lower bound on the delay 
of involution channels, whereas
it is easy to provide an unbounded SPF implementation. It hence 
follows that binary-valued circuit models based on involution channels 
allow to solve SPF precisely when this is possible 
in physical circuits. To the best of our knowledge, our model is hence the
very first candidate for a model that indeed guarantees faithful glitch 
propagation.
\end{abstract}


\section{Introduction}

The steadily increasing complexity of digital circuit designs in
     conjunction with the large simulation times of accurate analog
     simulations fuel the need for analysis techniques
     that are (i) fast and sufficiently accurate, and (ii) ideally
     also facilitate a formal analysis of circuit
     parameters/correctness at a sufficiently high level of
     abstraction.
Whereas there is a considerable body of work on timing analysis of circuits
     based on approximating the involved differential
     equations~\cite{NP73:spice,Ho84:thesis,LM84,PR90,DS90},
these approaches still suffer from large simulation times and 
high memory consumption.

Popular VHDL or Verilog simulators hence employ \emph{digital\/} timing
simulations, based on continuous-time, discrete-value, rather than 
analog-value, circuit models. Their modeling accuracy crucially depends on the
     ability to accurately predict the propagation of signal
     transitions throughout a circuit. More specifically,
precise timing models are not only important for accurate performance
and power consumption estimates at early design stages, but also 
for assessing a circuit's correctness: Bi-stable elements like latches, 
flip-flops, and arbiters fail to work correctly when glitches
or signal transitions occur at improper times, and may cause metastability
(including high-frequency pulse trains due to oscillatory 
metastability)~\cite{Mar81} on that occasion.
Since such phenomenons cannot simply be assumed to have vanished at
the occurrence of the next clock transition or the next handshake signal
in today's high-speed circuits,
the accurate prediction of the presence/absence of glitches and similar
short pulses is crucial.

Binary value, continuous time circuit models based on pure and
     inertial delay channels~\cite{Ung71} have been introduced
     several decades ago, and are still heavily used in existing
     digital design tools.
Those simple models cannot express such subtle phenomenons as decaying
     glitches, however: While pure delay channels propagate even very
short glitches as is, unlike real circuits,
     inertial delay channels make unrealistically strong
     assumptions~\cite{Mar77} by requiring a glitch to propagate unchanged
     when it exceeds some minimal length, and to completely vanish otherwise.
More elaborate digital channel models, like the PID model proposed by
     Bellido-D\'{\i}az et~al.\ \cite{BDJCAVH00}, have hence been
     introduced for building accurate digital timing analysis tools
     \cite{BJV06}.
Although the experimental validation of the PID model
     in~\cite{BDJCAVH00} showed good accuracy for the evaluated
     examples, the question of the general ability of such a model to
     actually capture the behavior of physical circuits remained open.

And indeed, F\"ugger et~al.~\cite{FNS13} showed that any model with
\emph{bounded single-history channels}, including pure delay, inertial delay,
and PID channels, fails to do so in the case of the simple \emph{Short-Pulse
Filtration\/} (SPF) problem: 
The SPF problem is the problem of building a one-shot variant of an inertial
delay channel. As for inertial delay channels, no short pulses may appear at
the SPF output; in case of long input pulses, however, they need not be passed
unaltered. In particular, the SPF output may also settle at logical~$1$ even if the 
input does not.
The stronger variant of \emph{bounded\/} SPF requires the SPF output to settle
in bounded time. 

Since Barros and Johnson~\cite{BJ83} proved that the
     problems of building an inertial delay, a latch, a synchronizer
     and an arbiter all are equivalent, the (un)solvability of 
(bounded) SPF is a suitable test for a model's ability to 
faithfully model glitch propagation with respect to physical circuits:
On the one hand, Marino~\cite{Mar77} formally proved
     that problems like SPF cannot be solved in a physical 
model when the output is required to
     stabilize in bounded time \cite{FNS13}.
On the other hand, a simple storage loop with a high-threshold filter
     at its output (see Fig.~\ref{fig:circuit}) solves SPF in
     unbounded time: As shown in the
SPICE simulation traces in Fig.~\ref{fig:spf}, sufficiently large 
input pulses (largest blue dashed one) just cause the storage loop to change 
its state (to 1) instantaneously (left-most green solid one), very small input 
pulses (smallest blue dashed one) don't affect the storage loop (bottom
green solid one). Critical input pulses (middle blue dashed ones, overlapping,
therefore appearing as if they were
one pulse) cause the storage
loop to become metastable for an unbounded time, eventually resolving 
to either state~0 or~1.
%
Therefore, appending a high threshold filter with threshold 
(marked by the red dotted line) clearly above the metastability 
region results in a clean (= non-metastable) output signal,
which either remains at~0, or makes a single transition to~1.
Hence, with real circuits, SPF is solvable, while its stronger bounded 
variant is not.

\newlength\figureheight
\newlength\figurewidth
\setlength\figureheight{4.8cm} 
\setlength\figurewidth{0.7\textwidth}
\ifthenelse{\boolean{PGFPlots}}{
\tikzsetnextfilename{spf_figure}
\tikzset{external/force remake}
}{}
\begin{figure}
  \centerline{
    \ifthenelse{\boolean{PGFPlots}}{
      \input{spf_figure.tex}
    }{
      \includegraphics[height=\figureheight]{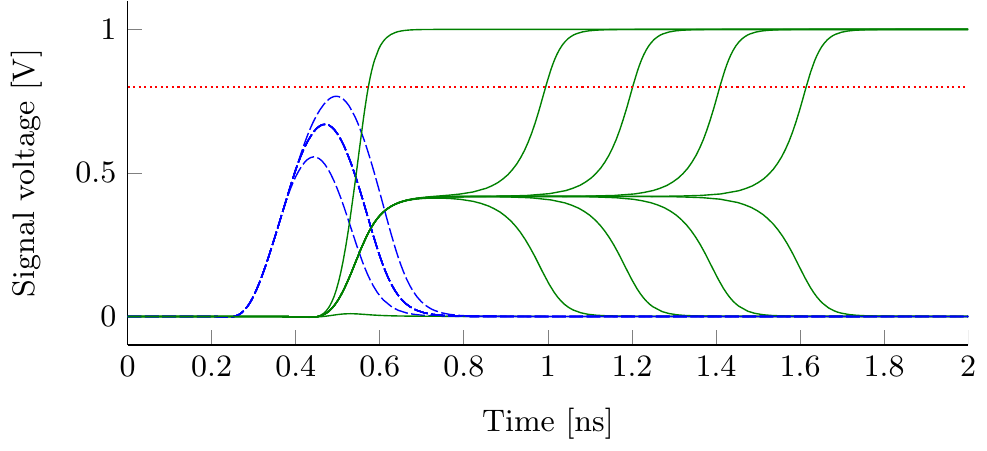}
    }}
  \caption{Analog simulation traces of a CMOS SPF, implemented as a storage 
loop followed by a high-threshold filter. The dashed (blue) curves represent
the input signal, the solid (green) ones give the output of the storage 
loop. The horizontal line at 0.8 marks the filter threshold level.}
  \label{fig:spf}
\end{figure}

\begin{figure}
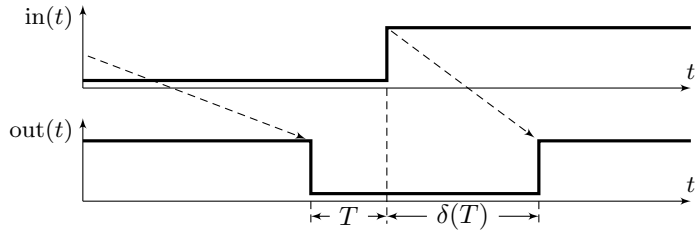

  \centerline{
    \tikzfigurechannelintro}
    \caption{
Input/output signal of a single-history channel, involving the
    input-to-previous-output delay $T$ and the resulting output-to-input delay
    $\delta(T)$}\label{fig:shc}
\end{figure}

A single-history channel, as introduced in \cite{FNS13}, is characterized by a
     delay function $\delta(T)$  that may depend on the
     difference~$T$ between the time of the input transition and that of
     the previous output transition.
\figref{fig:shc} illustrates this relation and the involved
     delays.
Pure delay, inertial delay, and PID channels are all single-history
     channels with an upper and lower {\em bounded\/} delay function.
Interestingly, as shown in~\cite{FNS13}, binary circuit models based on
     channels with pure (= constant) delays do not even allow to solve 
unbounded SPF. On the other hand, bounded single-history channels with 
non-constant delays, including inertial delay
     and PID channels, allow to design circuits
     that solve bounded SPF. Since this contradicts reality, as argued
above, none of the existing binary circuit models can faithfully
     capture glitch propagation.

In this paper, we propose a class of single-history channel models
      with {\em unbounded\/} delay functions: Like their bounded
      counterparts, their delay is upper bounded; however, it is not
      bounded from below.
As shown in Section~\ref{sec:channels}, these negative delays are
      crucial for accurately modeling glitch suppression.
We coined the term \emph{involution channel} for our channels, as
      we require their negative delay functions to be involutions, 
i.e., $-\delta(T)$ must form its own inverse (which implies that $\delta(T)$
is strictly increasing and concave).
To increase the size of our class of involution channels, we actually 
allow the delay functions~$\delta_\uparrow$ and~$\delta_\downarrow$ for rising and
      falling transitions to be different, and require both
$-\delta_\downarrow(-\delta_\uparrow(T)) = T$ and 
$-\delta_\uparrow(-\delta_\downarrow(T)) = T$. We will prove
that the solvability/unsolvability border of SPF
in a binary-valued circuit model based on our involution 
channels is exactly the same as in reality. It is hence,
to the best of our knowledge, the very first candidate for a model 
that indeed guarantees faithful glitch propagation.

\medskip

\noindent
{\bf Major contributions and paper organization:}
(1) In Section~\ref{sec:example}, we use a simple
analog channel model to demonstrate that assuming delay
functions which are involutions is not artificial
and hence not unrealistic: It reveals that the standard 
first-order model used e.g.\ in 
\cite{RCS90} actually gives a simple instance of general
involution channels, which are introduced formally in Section~\ref{sec:channels}.
Our binary circuit model, as well as the SPF problem, are formally defined in
Section~\ref{sec:model}.
(2) In Section~\ref{sec:possibility}, we prove that the simple circuit
consisting of a storage loop and a high-threshold filter
solves unbounded SPF in the involution channel model.
(3) In Section~\ref{sec:impossibility}, we show that
bounded SPF is impossible to solve with involution
channels. In a nutshell, our proof inductively constructs an
execution that can determine the final output only after some unbounded 
time. It exploits a surprising continuity property of the output 
of an involution channel with respect to the 
presence/absence of glitches at the channel input, which is due to
the involution property (unboundedness) of the delay functions.

Together, our results reveal that involution channels indeed allow
     to solve (bounded) SPF precisely when this is 
     possible in physical circuits, rendering them promising
     candidates for faithful glitch propagation models.
%

\medskip

\noindent{\bf Related Work.}
We are not aware of much existing work that relates to the problem
studied in our paper:
Unger~\cite{Ung71} proposed a general technique for modeling
asynchronous sequential switching circuits, based on 
combinational circuit elements interconnected by pure and 
inertial delay channels.
Brzozowski and Ebergen~\cite{BE92} formally proved that
it is impossible to implement Muller C-Elements and other
state-holding components using only zero-time logical gates 
interconnected by wires without timing restrictions.
Bellido-D\'\i az et~al.\ \cite{BDJCAVH00} proposed the PID model, and justified
its appropriateness both analytically and by comparing the model predictions
against SPICE simulations. However, as already mentioned, 
F\"ugger et~al.\ \cite{FNS13} 
showed that none of the above binary circuit models can faithfully 
model glitch propagation in physical circuits. 

\section{The Expressive Power of Involution Channels}
\label{sec:example}

Restricting delay functions to satisfy
the involution property $-\dup(-\ddo(T))=-\ddo(-\dup(T))=T$ 
might raise concerns about whether such an assumption makes 
sense at all in real circuits, and whether/how it fits to existing 
analog models~\cite{NP73:spice,Ho84:thesis,LM84,PR90,DS90}.
In this section, we will show that involution channels are indeed
well-suited for modeling physical circuits, in the sense that
they arise naturally in a (generalized) standard analog model.

More specifically, we will show that, for any given involutions
$\dup$, $\ddo$, there is a generalized standard analog 
channel model consisting of a pure delay component, a slew-rate limiter
with generalized switching waveforms, and a comparator, as shown 
in Fig.~\ref{fig:analogy}, which has
$\dup$, $\ddo$ as its corresponding delay functions. Note carefully, though, that
we do not claim that Fig.~\ref{fig:analogy} is the only analog model 
that leads to involution delay functions;
there may of course be many others as well. Vice versa, the fact
that some well-known analog model leads to involutions does 
not at all make our results incremental: Besides the fact that,
to the best of our knowledge, no analog modeling paper 
\cite{NP73:spice,Ho84:thesis,LM84,PR90,DS90} addressed the properties
of corresponding delay functions, it is of course not possible 
to generalize results obtained for some \emph{particular} involution
to involutions in general.

As a first observation, note that, while allowing separate functions $\dup$ for
rising and $\ddo$ for falling transitions, the timing behavior of involution 
channels is fully determined by either one, as
$\dup(T)=-\ddo^{-1}(-T)$ (and similarly for $\ddo$).
To better understand how our delay functions ``integrate'' the behavior of
both transitions, consider the ansatz
\begin{equation}\label{eq:ansatz}
\begin{split}
\dup(T)=-\fup^{-1}(\fdo(T)) \quad\text{ and }\quad \ddo(T)=-\fdo^{-1}(\fup(T)),
\end{split}
\end{equation}
where $\fup$ resp.\ $\fdo$ are strictly increasing resp.\ decreasing
functions.
Note that such functions can be found for any involution $\delta$ function.%
\footnote{One could choose $\fdo(t)=-t$ and $\fup(t)=\ddo(t)$, for example.}
Intuitively, we would like $\fup$ and $\fdo$ to represent the continuous switching waveforms
of the output of the generalized slew rate limiter upon the occurrence of a rising
respectively falling transition at its input.
In the above formula, e.g., at a rising transition, $\dup(T)$ returns the time by
which $\fup$ has to be shifted so that the output signal remains continuous with
respect to the output caused by the previous falling transition.
%
For realistic switching waveforms, we further need $\fup(0)=1-\fdo(0)=0$
and $\lim_{t\to\infty}\fup(t)=1-\lim_{t\to\infty}\fdo(t)=1$,\footnote{Still,
  any $\delta$ can be constructed in this way,
  e.g., by using~$\fdo(t)=e^{-t}$,
  $\fup(t)=e^{\ddo(t-\delta^\uparrow_\infty)-\delta^\downarrow_\infty}$.} which requires
to augment (\ref{eq:ansatz}) with some additive terms, resulting in
\begin{equation}\label{eq:ansatz2}
\begin{split}
\dup(T)=-\fup^{-1}(\fdo(T+\delta^\downarrow_\infty))+\delta^\uparrow_\infty \text{~~and~~}
\ddo(T)=-\fdo^{-1}(\fup(T+\delta^\uparrow_\infty))+\delta^\downarrow_\infty,
\end{split}
\end{equation}
where $\delta^\uparrow_\infty$ and $\delta^\downarrow_\infty$ denote
$\lim_{T\to\infty}\dup(T)$ and $\lim_{T\to\infty}\ddo(T)$, respectively.

\begin{figure}[tb]
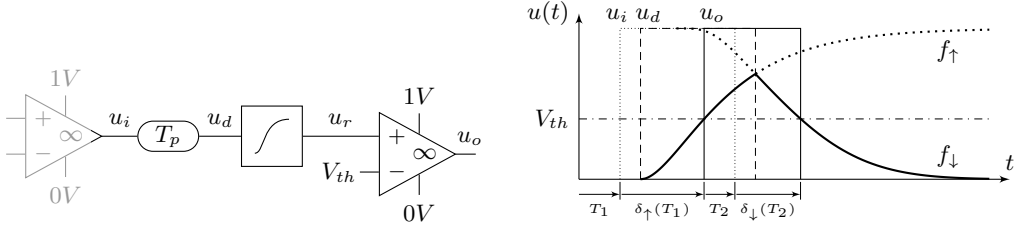

  \centerline{
    \tikzfigurecircuit~\tikzfigureanalogy}
  \caption{Simple analog channel model.}
  \label{fig:analogy}
\end{figure}
Fig.~\ref{fig:analogy} shows a block diagram of an idealized analog circuit
corresponding to so constructed involution channels, and a sample waveform.
The pure delay shifts the binary-valued input $u_i$ in time by some~$T_p$.
The slew rate limiter exchanges the step functions of the resulting $u_d$ with
instances of $\fup$ and $\fdo$, shifting them in time such that the output $u_r$
is continuous and switches between strictly increasing and decreasing exactly at
$u_d$ switching times.
The comparator generates $u_o$ by again discretizing the value of this waveform
comparing it to the threshold voltage $V_{th}$, effectively adding
$\fup^{-1}(V_{th})$ resp.\ $\fdo^{-1}(V_{th})$ to the
instantiation times of $\fup$ resp.\ $\fdo$.
The input-output delay of a perfectly idle channel (the last output transition
was at time $-\infty$), i.e.~$\delta^\uparrow_\infty$ and
$\delta^\downarrow_\infty$ for rising respectively falling transitions, is the
sum of the pure delay and the time the switching waveform needs to reach the
threshold voltage $V_{th}$; e.g.,~for a rising transition,
$\delta^\uparrow_\infty= T_p + \fup^{-1}(V_{th})$.

Apart from showing that, for any $\delta$ function, there is a combination of pure
delay $T_p$, switching waveforms $\fup$ and $\fdo$, and threshold $V_{th}$ so
that the circuit in Fig.~\ref{fig:analogy} behaves exactly like the corresponding involution channel,
(\ref{eq:ansatz2}) can also be used to directly transform the parameters of the
model in Fig.~\ref{fig:analogy} to the corresponding $\delta$ function.
As a special case, consider a slew rate limiter implemented as a first-order
RC low pass filter; the switching waveforms are
$\fdo(t)=1-\fup(t)=e^{-t/\tau}$ here, with $\tau$ being the RC time constant.
Inserting these functions and their inverses into (\ref{eq:ansatz2}) and
substituting $\delta^\uparrow_\infty$ and $\delta^\downarrow_\infty$ with the
corresponding sums of pure delay and comparator delay, we obtain
\begin{equation*}
\begin{split}
\delta_\uparrow(T)&=\tau\ln(1-e^{-(T+T_p-\tau\ln(V_{th}))/\tau})+T_p-\tau\ln(1-V_{th})\\
\delta_\downarrow(T)&=\tau\ln(1-e^{-(T+T_p-\tau\ln(1-V_{th}))/\tau})+T_p-\tau\ln(V_{th}).
\end{split}
\end{equation*}
In the remainder of this paper, these specific (well-known) channels will be called
\emph{exp-channels}.

\ifthenelse{\boolean{true}}{}{
In order to motivate why involutions are promising
candidates for suitable $\delta$-functions, consider the simple analog channel model
depicted in \figref{fig:analogy}. 
\begin{figure}
  \centerline{
    \tikzfigurecircuit}
  \centerline{
    \tikzfigureanalogy}
  \caption{Simple analog channel model.}
  \label{fig:analogy}
\end{figure}
This well-known model, see e.g.\ \cite{RCS90} for an instance, consists of a (pure) 
delay element with delay $T_p$, a slew rate limiter
and a comparator, 
all of which are idealized. %
The circuit input $u_i$, coming from the comparator of the previous stage, 
hence takes on the value 0 or 1 [Volt] 
and switches between those two values with infinite slope. %
The unit [Volt] will be omitted subsequently.
Both $u_i(t)$ and the output $u_d(t)$ of the delay element can hence be viewed
as binary-valued signals.
The slew rate limiter replaces the infinite-slope transitions of $u_d$ with the
predefined switching waveforms $f_\uparrow$ for the rising edge and
$1-f_\downarrow$ for the falling edge on its output $u_r$.
copy waveform These functions (collectively termed ``$f$'' below) must have the
following properties: %
$f(0) = 0$, $\lim_{t\to\infty} f(t) = 1$, and $f$ is strictly increasing and
continuous.
Finally, the comparator discretizes $u_r$ by comparing its value to
a threshold $V_{th}$, thereby generating output signal~$u_o$. %

In order to analyze the behavior of such a channel, we consider
its input $u_i(t)$ to be a signal made up of a sequence of alternating
transitions $(t_n,x_n)_{n\geq0}$: $t_n$ is the time of the $n$-th transition,
and $x_n=0$ resp.\ $x_n=1$ identifies it to be falling resp.\ rising;
the initial transition occurs at time $t_0=-\infty$ for convenience.
Let $(\hat{t}_n)_{n\geq0}$ be the corresponding sequence of switching times of the pure
delay output $u_d$,
i.e., $\hat{t}_n = t_n+T_p$ for $n\geq 0$. Note that $\hat{t}_0=-\infty$,
and assume for simplicity that the initial values of $u_i$ and 
$u_d$ are $x_0=\hat{x}_0=0$.
Then, $\forall k \in
\mathbb{N}_0, \forall t \in [\hat{t}_{2k}, \hat{t}_{2k+1}): u_d(t) = 0$ and similarly
$u_d(t) = 1$ when $t \in [\hat{t}_{2k+1}, \hat{t}_{2k+2})$. %
The slew rate limiter replaces the infinite-slope transitions with instances of
$f_\uparrow$ and $1-f_\downarrow$ as follows: 
$\forall k \in \mathbb{N}_0, \forall t \in [\hat{t}_{2k},
\hat{t}_{2k+1}): u_r(t) = 1-f_\downarrow(t-\hat{t}_{2k}+\theta_{2k})$ and similarly
$u_r(t) = f_\uparrow(t-\hat{t}_{2k+1}+\theta_{2k+1})$ when $t \in [\hat{t}_{2k+1},
\hat{t}_{2k+2})$. 
The sequence $(\theta_n)_{n\geq0}$ is defined implicitly, by requesting
continuity of $u_r(t)$ for all $t$. More explicitly, this requires
$\theta_0=0$ and, in case of $n=2k+1$, i.e., a rising transition at $\hat{t}_n$, 
$\theta_n = f_\uparrow^{-1}(u_r(\hat{t}_n))$. Substituting $u_r(\hat{t}_n)$, we get $\theta_n =
f_\uparrow^{-1} (1-f_\downarrow(\hat{t}_n-\hat{t}_{n-1}+\theta_{n-1}))$. %
For falling transitions, the formula is the same with $f_\uparrow$ and
$f_\downarrow$ swapped; note that the inverse of $1-f(x)$ is just
$f^{-1}(1-x)$. %

The comparator again produces the ``binary-valued'' output signal $u_o(t)$, by
comparing $u_r$ to a threshold voltage $V_{th} \in (0,1)$. %
Knowing that $u_r$ is composed of alternating instances of the bijective functions
$f_\uparrow$ and
$1-f_\downarrow$, there exist unique
$\Delta_\uparrow = f_\uparrow^{-1}(V_{th})$ and $\Delta_\downarrow =
f_\downarrow^{-1}(1-V_{th})$ such that $f_\uparrow(\Delta_\uparrow) = V_{th}$
and $1-f_\downarrow(\Delta_\downarrow) = V_{th}$. %
Therefore, we can derive $u_o(t)$ directly from 
$u_r$, by generating rising transitions at time $t_{2k+1}-\theta_{2k+1}+\Delta_\uparrow$
and falling transitions at $t_{2k}-\theta_{2k}+\Delta_\downarrow$.
Note carefully that the resulting output transition times need not be 
strictly monotonically increasing any more, which results in the cancellation
of transitions in a single-history channel already introduced in \cite{FNS13}.

The overall input to output behavior of the channel for any rising
input transition $(t_n,1)$ on $u_i$ can now be stated as follows: %
$\hat{t}_n=t_n + T_p$ is mapped to the instance 
$f_\uparrow(t-t_n-T_p+\theta_n)$ in the slew rate limiter, from which
the comparator generates the corresponding transition of $u_o$
at time $t_n'=t_n+T_p-\theta_n +
f_\uparrow^{-1}(V_{th})$. %
Substituting for $\theta_n$, we get $t_n'=t_n+T_p-f_\uparrow^{-1}
(1-f_\downarrow(t_n-t_{n-1}+\theta_{n-1})) + f_\uparrow^{-1}(V_{th})$. %
Utilizing the input-to-previous-output transition time
$T=t_n-t_{n-1}-T_p+\theta_{n-1} - f_\downarrow^{-1}(1-V_{th})$, we obtain
$t_n'=t_n+T_p-f_\uparrow^{-1} (1-f_\downarrow(T+T_p+f_\downarrow^{-1}(1-V_{th}))) +
f_\uparrow^{-1}(V_{th})$. The output-to-input delay $\delta_\uparrow(T)=t_n'-t_n$ (and
$\delta_\downarrow(T)$, which is obtained analogously) is thus: %
\[
{\small
\begin{split}
\delta_\uparrow(T) = T_p-f_\uparrow^{-1} (1-f_\downarrow(T+T_p+f_\downarrow^{-1}(1-V_{th}))) +
f_\uparrow^{-1}(V_{th})
\\
\delta_\downarrow(T) = T_p-f_\downarrow^{-1} (1-f_\uparrow(T+T_p+f_\uparrow^{-1}(V_{th}))) +
f_\downarrow^{-1}(1-V_{th})
\end{split}
}
\]
\normalsize

If $V_{th}=0.5$ is plugged into the definitions above, we obtain
$\delta_\uparrow(T) = \delta_\downarrow(T)=\delta(T)$. By plugging in $T:=-\delta(T)$
into the resulting definition of $\delta(T)$, it is easy to verify that $-\delta(T)$ is
indeed an involution.

We conclude this section with the observation that the model used
in \cite{RCS90}, which uses a first-order RC low-pass filter for
the slew rate limiter, is actually a particular simple instance 
of an involution channel: It produces transitions with $f_\uparrow(t) = f_\downarrow(t) = f(t) =
1-e^{-(t/\tau)}$, where $\tau$ is the RC constant. %
The inverse is $f^{-1}(u) = -\tau \ln(1-u)$, which leads to the following
$\delta$-functions:
{\small
\begin{eqnarray*}
\delta_\uparrow(T) &=& \tau\ln(1-e^{-(T+T_p-\tau\ln(V_{th}))/\tau})+T_p-\tau\ln(1-V_{th})\\
\delta_\downarrow(T) &=& \tau\ln(1-e^{-(T+T_p-\tau\ln(1-V_{th}))/\tau})+T_p-\tau\ln(V_{th})
\end{eqnarray*}
}\normalsize
In the sequel, these channels will be called \emph{exp-channels}.

} 

\section{Binary Circuit Model}\label{sec:model}

Since the purpose of our work is to replace analog models like the one 
in the previous section by a purely digital model, we will now formally
define the binary-value continuous-time circuit model used in the remainder of
this paper. Except for the involution channels introduced in 
Section~\ref{sec:channels}, it is essentially the same as the model introduced
in~\cite{FNS13}.

\smallskip

\bfno{Signals.}
A {\em falling transition\/} at time~$t$ is the pair $(t,0)$, a {\em rising
transition\/} at time~$t$ is the pair $(t,1)$.
A {\em signal\/} is a (finite or infinite) list of alternating transitions
such that
\begin{enumerate}[S1)]
\item the initial transition is at time~$-\infty$; all other transitions are at times~$t\geq0$,
\item the sequence of transition times is strictly increasing,
\item if there are infinitely many transitions in the list, then the set of
transition times is unbounded.
\end{enumerate}

To every signal~$s$ corresponds a function $\IR_+\to\{0,1\}$ whose value at
time~$t$ is that of the most recent transition.
We follow the convention that the function already has the new value at the
time of a transition, i.e., the function is constant in the half-open interval
$[t_n,t_{n+1})$ if~$t_n$ and~$t_{n+1}$ are two consecutive
transition times.
A signal is uniquely determined by such a function and its value at~$-\infty$.

\smallskip

\bfno{Circuits.}
Circuits are obtained by interconnecting a set of input ports and a set
     of output ports, forming the external interface of a circuit, 
and a set of combinational gates via channels. 
We constrain the way components are interconnected in a natural
way, by requiring that any gate input, channel input and output
port is attached to only one input port, gate output or channel
output. Moreover, gates and channels must alternate on every
path in the circuit.

Formally, a {\em circuit\/} is described by a directed graph where:
\begin{enumerate}[C1)]
\item Vertices are partitioned
into {\em input ports}, {\em output ports}, {\em channels}, and {\em gates}.
\item Input ports have no incoming edges and at least one outgoing edge.
\item Output ports have exactly one incoming edge from a gate and no outgoing edges.
\item Channels are nodes that have exactly one incoming and exactly one outgoing
edge. Every channel is assigned a channel function, which maps the input 
to the output. Section~\ref{sec:channels} specifies the properties 
of this function for our involution channels.
\item 
Every gate is assigned a Boolean function
      $\{0,1\}^d\to\{0,1\}$, where~$d$ is the number of incoming edges.
\item
There is a fixed order on the incoming edges of every gate.
\item
Gates and channels alternate on every path in a circuit.
\end{enumerate}


\smallskip

\bfno{Executions.}
An execution of circuit~$C$ is an assignment of signals to vertices that
respects the channel functions and Boolean gate functions.

Formally, an {\em execution\/} of circuit~$C$ is a collection of signals~$s_v$
for all vertices~$v$ of~$C$ such that the following properties hold:
\begin{enumerate}[E1)]
\item
If~$i$ is an input port, then there are no restrictions on~$s_i$.
\item
If~$o$ is an output port, then~$s_o = s_v$, where~$v$ is the unique
gate $v$ associated with~$o$.
\item
If~$c$ is a channel, then $s_c = f_c(s_v)$, where~$v$ is the unique incoming
neighbor of~$c$ and~$f_c$ the channel function.
\item
If~$b$ is a gate with~$d$ incoming neighbors $v_1,\dots,v_d$, ordered according to
the fixed order of condition (C6), and gate function~$f_b$, then for all times~$t$,
\[
s_b(t) = f_b\big( s_{v_1}(t) ,  s_{v_2}(t) , \dots , s_{v_d}(t) \big)\enspace.
\]
\end{enumerate}


%

\smallskip

\bfno{Short-Pulse Filtration.} 
A {\em pulse\/} of length~$\Delta$ at time~$T$ has initial value~$0$, one
rising transition at time~$T$, and one falling transition at time $T+\Delta$.
A signal {\em contains a pulse\/} of length~$\Delta$ at time~$T$ if it
 contains a rising transition at time~$T$, a falling transition at time
$T+\Delta$ and no transition in between.

A circuit {\em solves Short-Pulse Filtration (SPF)\/} if it fulfills the
following conditions. Note that we allow the circuit to behave arbitrarily if the input
signal is not a (single) pulse.
\begin{enumerate}[F1)]
\item The circuit has exactly one input port and exactly one output port. {\em
        (Well-formedness)}
\item If the input signal is the zero signal, then so is the output
	signal. {\em (No generation)}
\item There exist an input pulse such that
	the output signal is not the zero signal. {\em (Nontriviality)}
\item There exists an~$\varepsilon>0$ such that for every input pulse the output
	signal never contains a pulse of length less than~$\varepsilon$. {\em
	(No short pulses)}
\end{enumerate}

A circuit {\em solves bounded SPF\/} if additionally the following condition holds:
\begin{enumerate}
\item[F5)] There exists a $K>0$ such that for every input pulse
the last output transition
	is before time~$T+K$ if~$T$ is the time of the last input transition.
	{\em (Bounded stabilization time)}
\end{enumerate}

\section{Involution Channels}\label{sec:channels}

Intuitively, a channel propagates each transition at time~$t$ of the
input signal to a transition at the output happening after some \emph{output-to-input}
delay $\delta(T)$, which depends on the \emph{input-to-previous-output}
delay $T$. 
Note that~$T$ can be negative if two input transitions are close together,
as is the case in Fig.~\ref{fig:shc:T:neg}.
\begin{figure}
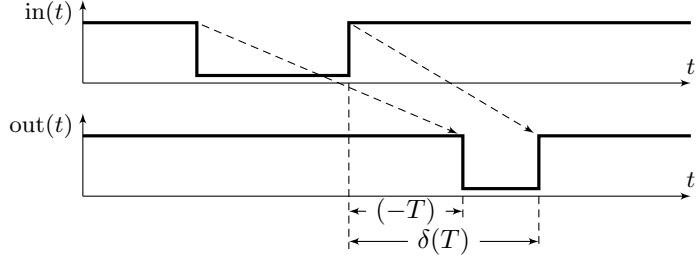

  \centerline{
    \tikzfigurechannelsection}
    \caption{
Input transition with negative input-to-last-output delay $T$
    }\label{fig:shc:T:neg}
\end{figure}


Formally, an {\em involution channel\/} is characterized by an {\em initial
     value\/} $I\in\{0,1\}$ and
     two strictly increasing concave {\em delay
functions\/} 
$\delta_\uparrow:(-\delta^\downarrow_\infty,\infty)\to(-\infty,\delta^\uparrow_\infty)$
and
$\delta_\downarrow:(-\delta^\uparrow_\infty,\infty)\to(-\infty,\delta^\downarrow_\infty)$
such that both 
$\delta^\uparrow_\infty=\lim_{T\to\infty}\delta_\uparrow(T)$ 
and
$\delta^\downarrow_\infty=\lim_{T\to\infty}\delta_\downarrow(T)$ 
are finite and
\begin{equation}\label{eq:involution}
-\delta_\uparrow\big( -\delta_\downarrow(T) \big) = T
\text{ and }
-\delta_\downarrow\big( -\delta_\uparrow(T) \big) = T
\end{equation}
for all applicable~$T$.
All such functions are necessarily continuous and strictly increasing.
{For simplicity, we will also assume them to be differentiable}; $\delta$
being concave thus implies that its derivative $\delta'$ is monotonically
decreasing.
If multiple channels in a circuit share a common input signal, as depicted in
\figref{fig:branch}, we require that they all have the same initial value $I$.
This is without loss of generality, as one can always replicate the input
signal.


\begin{figure}
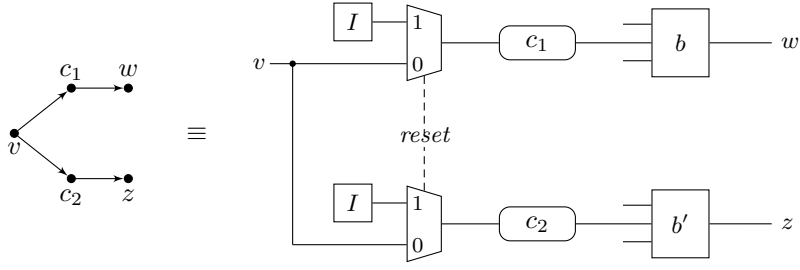

  \centerline{
  \tikzfigurechannelbranch}
\caption{A circuit (graph) with vertex~$v$ (being an input or a gate),
     gates~$w$, $z$, and channels~$c_1$
     and~$c_2$ (on the left) and the physical equivalent (on the right).
Both channels must have the same initial value $I$; $b$ and~$b'$ are the
     Boolean functions assigned to gates~$w$ and~$z$,
     respectively.}\label{fig:branch}
\end{figure}

The behavior of involution channels is defined as
follows:

\noindent
{\em Initialization:\/} 
If the channel's initial value~$I$ is different from the initial value~$X$ of
the channel input signal~$s$ and~$s$ has no transition at time~$0$, add the
transition $(0,X)$ at time~$0$ to~$s$ (``reset'').

\noindent
{\em Output transition generation algorithm:\/}
Let~$t_1,t_2,\dots$ be the times of the transitions of~$s$, and
set $t_0 = -\infty$ and $\delta_0=0$.
\begin{itemize}
\item {\em Iteration:\/} Determine the tentative list of \emph{pending\/} output
transitions:
Recursively determine the output-to-input delay $\delta_n$ for the input transition at time~$t_n$ by
setting
$\delta_n = \delta_\uparrow(t_n - t_{n-1} - \delta_{n-1})$ if $t_n$ is a rising
transition and
$\delta_n = \delta_\downarrow(t_n-t_{n-1}-\delta_{n-1})$
if it is falling.
The $n$th and $m$th pending output transitions {\em cancel\/} if $n < m$ but
$t_n+\delta_n \geq t_m+\delta_m$. In this case, we mark both as canceled.

\item {\em Return:\/}
The channel output signal~$c(s)$ has initial value~$I$ and contains every pending 
transition at time $t_n + \delta_n$ that has not been marked as canceled.

\end{itemize}

\begin{definition}
An involution channel is {\em strictly causal\/} if $\delta_\uparrow(0)>0$,
which is equivalent to the condition $\delta_\downarrow(0)>0$ due to (\ref{eq:involution}).
\end{definition}

\begin{lemma}
An exp-channel is strictly causal if and only if $T_p >0$.
\end{lemma}

The next lemma\ifthenelse{\boolean{SHORTversion}}{\footnote{Due to lack of
		space, the detailed proofs of the full paper have been
relegated to Appendix~\ref{sec:proofs}.}}{}
identifies an important parameter $\delta_{\min}$ of 
a strictly causal involution channel, which gives its minimal pure delay.

\begin{lemma}
\label{lem:delta:min}
A strictly causal involution channel has a unique~$\delta_{\min}$ defined
by $\delta_\uparrow(-\delta_{\min}) = \delta_{\min} = 
\delta_\downarrow(-\delta_{\min})$, which is positive.
For exp-channels, $\delta_{\min}=T_p$.

For the derivative, we have $\delta_\uparrow'(-\delta_\downarrow(T)) = 1/\delta_\downarrow'(T)$
and hence $\delta_\uparrow'(-\delta_{\min})  =
1/\delta_\downarrow'(-\delta_{\min})$.
\end{lemma}
\ifthenelse{\boolean{SHORTversion}}{}{
\begin{proof}
Set $f(T) = -T + \delta_\uparrow(-T)$.
This function is continuous and strictly decreasing, 
since~$\delta_\uparrow$ is continuous and strictly increasing.
Because $f(0)=\delta_\uparrow(0)$ is positive and the limit of
$f(T)$ as
$T\to
\delta^\downarrow_\infty$ is $-\infty$, there exists a unique~$\delta_{\min}$
between~$0$
and~$\delta^\downarrow_\infty$ for which $f(\delta_{\min}) = 0$.
Hence, $\delta_\uparrow(-\delta_{\min})=\delta_{\min}$.
The second equality follows from 
$\delta_{\min}=\delta_\downarrow(-\delta_\uparrow(-\delta_{\min}))=\delta_\downarrow(-\delta_{\min})$
according to (\ref{eq:involution}).

The second part of the lemma follows by differentiating
Equation~\eqref{eq:involution}.
\end{proof}
}

We next show that~$\delta_{\min}$ indeed deserves its name:
A particular consequence of the following lemma is that the channel delay
for any non-canceled transition is at least~$\delta_{\min}$.

\begin{lemma}
\label{lem:minimal:delay}
The $n$th and $(n+1)$th pending output transitions cancel if and only if 
$t_{n+1} \leq t_n + \delta_n - \delta_{\min}$.
\end{lemma}
\ifthenelse{\boolean{SHORTversion}}{}{
\begin{proof}
Let~$\delta$ be either~$\delta_\uparrow$ or~$\delta_\downarrow$, depending on
whether $t_{n+1}$ is a rising or falling transition.
By definition, the two transitions cancel if and only if
\begin{equation}\label{eq:lem:minimal:delay:proof}
\delta_{n+1} = \delta(t_{n+1}-t_n-\delta_n)
\leq
-(t_{n+1} - t_n - \delta_n)
\enspace.
\end{equation}
Set $T = t_{n+1} - t_n - \delta_n$.
By Lemma~\ref{lem:delta:min}, equality holds
in~\eqref{eq:lem:minimal:delay:proof} if and only if $T = -\delta_{\min}$.
Because the left-hand side of~\eqref{eq:lem:minimal:delay:proof} is increasing
in~$T$ and the right-hand side is strictly decreasing in~$T$,
\eqref{eq:lem:minimal:delay:proof} is equivalent to $T\leq -\delta_{\min}$,
which in turn is equivalent to $t_{n+1}\leq t_n + \delta_n - \delta_{\min}$.
\end{proof}
}

In the rest of the paper, we assume all channels to be strictly causal involution
channels.

\section{Constructing Executions of Circuits}\label{sec:algorithm}

The definition of an execution of a circuit as given 
in Section~\ref{sec:model} is ``existential'', in the sense that it
only allows to check for a given collection of signals whether 
it is an execution or not. And indeed, in general, circuits may 
have no execution or may have several different executions. 
By contrast, in case of circuits involving strictly causal involution 
channels only, 
executions are unique and can be constructed iteratively: 
We give a deterministic construction algorithm below.

Given a circuit~$C$ with strictly causal involution channels, let 
$(s_i)_{i\in \mathcal{I}}$ be any collection of signals for all the input 
ports $\mathcal{I}$. 
Since all output ports are driven by gates we can identify the output port with
the output of its driving gate.
The channel with predecessor~$x$ (an input port or a gate output) and
successor~$y$ (a gate input) is denoted by the tuple~$(x,y)$.
The algorithm iteratively generates the list of transitions of $s_\sigma$ of
(the output of) every vertex $\sigma$ in the circuit, and hence the
corresponding function $s_\sigma(t)$.
In the course of the execution of this algorithm, a subset of the generated
transitions will be marked {\em fixed}: Non-fixed transitions could still be
canceled by other transitions later on, fixed transitions will actually occur in
the constructed execution.

\medskip

The detailed algorithm is as follows:          

\noindent
{\em Initialization:\/} 
For all channels~$(v,w)$ in~$C$, $s_{(v,w)} = ((-\infty,I))$ initially, with~$I$ being the
     initial value of channel~$(v,w)$. According to the implicit reset of our
channels introduced in Section~\ref{sec:channels}, the transition $(0,X)$ is
also added to $s_{(v,w)}$ if the initial transition $(-\infty,X)$ of $s_v$
satisfies $X\neq I$.\footnote{Note that this is well-defined also in case of
channels~$(v,w)$ and~$(v,w')$ attached to the same~$v$, as we require 
$s_{(v,w)}=s_{(v,w')}$ initially in this case; see Section~\ref{sec:channels}.} 
For a gate~$v$, $s_v =((-\infty,X))$ initially,
     where~$X$ is the value of the Boolean function corresponding to~$v$
     applied to the values of the initial transitions in $s_{\sigma}$ 
     for all of~$v$'s predecessors~$\sigma$. The zero-input gates $0$
and $1$ used for generating constant-0 and constant-1 signals have
$s_0=((-\infty,0))$ and $s_1=((-\infty,1))$, respectively.
Initially, all transitions at~$-\infty$ are fixed and all others are not.

\noindent
{\em Iteration:\/} If there is no non-fixed transition left, terminate with
the execution made up by all fixed transitions. Otherwise, let~$t \geq 0$
     be the smallest time of a non-fixed transition.
\begin{enumerate}
\item[(i)] Mark all transitions at~$t$ fixed.
\item[(ii)] For each newly fixed transition from step~(i), occurring in $s_{\sigma}$ 
where~$\sigma$ is a predecessor of a gate~$v$: If signal~$s_v$'s current value $s_v(t)=X$
differs from the value of~$v$'s Boolean function applied to the values $s_{\sigma'}(t)$
for all of~$v$'s predecessors~$\sigma'$ (which also include $\sigma$), add the
transition~$(t,1-X)$ to~$s_v$ and mark it fixed.
\item[(iii)] For each newly fixed transition $(t,x)\in s_v$ from steps~(i) or~(ii), occurring
in $s_v$ of a gate output or an input port: For each successor channel~$(v,w)$ of~$v$, apply the
     iteration step of~$(v,w)$'s transition generation algorithm with input
     signal~$s_v$, output signal~$s_{(v,w)}$, and current input transition $(t,X)$.
If this leads to a cancellation in $s_{(v,w)}$, remove both canceling and canceled
transition from the list. Lemma~\ref{lem:fixed_is_fixed} will show
that no fixed transition will ever be removed this way.
\end{enumerate}

We will now show that this algorithm indeed constructs an execution of $C$.
Let~$t_\ell$ be the smallest finite time of non-fixed transitions at the 
beginning of iteration~$\ell \ge 1$ of the algorithm, and denote 
by~$\delta^C_{\min}>0$ the minimal $\delta_{\min}$ of all channels
     in circuit~$C$.

\begin{lemma}\label{lem:t_fixed}
For all iterations~$\ell \ge 1$, (a) no transition $(t,X)$ with
     $t\neq t_\ell$ is newly marked fixed in the iteration, (b) a
     transition $(t,X)$ added during and not removed by the end of iteration~$\ell$
     either has time~$t=t_\ell$ or $t > t_{\ell}+\delta_{\min}^C > t_{\ell}$, and (c)
     every transition at time~$t_\ell$ is fixed at the end of the
     iteration.
\end{lemma}
\ifthenelse{\boolean{SHORTversion}}{}{
\begin{proof}
Statement~(a) is implied by the fact that transitions are only marked fixed in
     step~(i) and~(ii), which act on transitions at time~$t_\ell$ only.

For~(b), assume by contradiction that a
     transition $(t,X)$ with $t \leq
t_\ell + \delta_{\min}^C$ but different from~$t_\ell$ was added in
     iteration~$\ell$ and still exists at the end of iteration~$\ell$.
Such a transition can only be added via step~(iii).
For the respective channel algorithm with delay function~$\delta$,
     $\delta(t_\ell-t') \leq \delta_{\min}^C$ must have held, where~$t'$ is
     the time of the channel's last output transition.
From Lemma~\ref{lem:delta:min}, we deduce that this implies 
 $t_\ell \leq t'-\delta_{\min}$ for the particular channel's minimal delay $\delta_{\min}$, 
since $\delta_{\min}^C\leq\delta_{\min}$. By Lemma~\ref{lem:minimal:delay},
this leads to a cancellation and hence removal of $(t,X)$, which provides the required 
contradiction.

For~(c), assume by contradiction that, at the end of
     iteration~$\ell$, there exists a non-fixed transition $(t_\ell,X)$.
Since step~(i) marks all transitions at time~$t_\ell$ fixed and~(ii) adds
     only fixed transitions at time~$t_\ell$, the non-fixed transition must have
     been newly added in step~(iii).
However, from~(b), we know that this requires $t > t_\ell + \delta_{\min}^C >
t_\ell$, a contradiction.
\end{proof}
}

From an inductive application of Lemma~\ref{lem:t_fixed}, we
obtain that the sequence of iteration start times $(t_\ell)_{\ell \ge 1}$ 
is strictly increasing without bound:

\begin{lemma}\label{lem:t_inc}
For all iterations $\ell > 1$, $t_{\ell}-t_{\ell-1}>0$. If $t_{\ell}$ does not
involve an input transition, then $t_{\ell}-t_{\ell-1} > \delta_{\min}^C$.
\end{lemma}
\ifthenelse{\boolean{SHORTversion}}{}{
\begin{proof}
By Lemma~\ref{lem:t_fixed}~(b), $t_{\ell+1}$ is larger than
$t_{\ell}+\delta_{\min}^C$, provided no input transition occurs earlier.
As we do not allow Zeno behavior of input signals, $t_{\ell}-t_{\ell-1}>0$
is guaranteed also in the latter case.
\end{proof}
}

The following lemma proves that the generated event lists are well-defined,
in the sense that no later iteration can remove events that may have generated
causally dependent other events already.

\begin{lemma}\label{lem:fixed_is_fixed}
No fixed transition is canceled in any iteration.
\end{lemma}
\ifthenelse{\boolean{SHORTversion}}{}{
\begin{proof}
Assume by contradiction that some iteration~$\ell \ge 2$ is the first
     in which a fixed transition is canceled; obviously, this can only
happen in step~(iii).
Thus, there exists a transition at time~$t_\ell$ that generated a new transition
     at some time~$t$ that results in the cancellation of a fixed transition at
     time~$t'$, i.e., $t\leq t'$. 
Lemma~\ref{lem:minimal:delay} implies that $t_{\ell}-t' \leq -\delta_{\min} <0$
     in this case.
By Lemma~\ref{lem:t_fixed}.(a)--(c), however, $t \leq t' \leq t_\ell$ and thus $t_{\ell}-t' \geq 0$,
which provides the required contradiction.
\end{proof}
}

We are now ready for the main result of this section, which asserts the
existence of a unique execution of our circuit~$C$:

\begin{theorem}\label{thm:execution}
The execution construction algorithm either terminates or, for all times
$T\geq0$, there exists an iteration~$\ell$ such that $t_\ell \geq T$.
At the end of iteration~$\ell \ge 1$, the collection of signals~$s_\sigma$, restricted to time~$[-\infty,t_\ell]$,
is the unique execution of circuit~$C$ restricted to time~$[-\infty,t_\ell]$.
If the algorithm terminates at the beginning of iteration~$\ell$, 
then this collection of signals is the unique execution of circuit~$C$.
\end{theorem}
\ifthenelse{\boolean{SHORTversion}}{}{
\begin{proof}
From Lemma~\ref{lem:t_inc}, we deduce that for all times~$t\ge 0$,
     there is an iteration~$\ell \ge 1$ such that~$t_\ell > t$ or the
     algorithm terminates.
From Lemma~\ref{lem:fixed_is_fixed}, we further know that in both
     cases the algorithm does not add transitions with times less or equal
     to~$t$.
Uniqueness of the execution follows from the fact that the construction
algorithm is deterministic.
\end{proof}
}

\section{Possibility of Unbounded Short-Pulse Filtration}\label{sec:possibility}

In this section, we show that unbounded SPF is solvable in our circuit model
     with strictly causal involution channels.
We do this by verifying that the circuit shown in \figref{fig:circuit}
     indeed solves SPF.
The circuit was inspired by the physical solution of~\figref{fig:spf},
     and consists of a fed back OR-gate forming the storage loop and a 
     subsequent high-threshold filter (implemented by 
     a channel).
In order not to obfuscate the essentials (and to stick to the page limit),
     we restrict
     \footnote{However, the proof could 
        be adapted to show the possibility of unbounded SPF for many classes
        of strictly causal involution channels.}
     our attention to certain classes of involution channels.
More specifically, in our proof,
the channel in the feed-back loop must be strictly causal and symmetric,
i.e., $\delta_\uparrow=\delta_\downarrow=\delta$. When using an
exp-channel, for example, this implies a threshold~$V_{th}=0.5$.
The channel implementing the high-threshold filter is
assumed to be an exp-channel because we have to adjust its parameters
appropriately. 

\begin{figure}
\centering
\begin{tikzpicture}[>=latex',circuit logic US, scale=1.3]
\matrix[column sep=0mm,row sep=2mm,every node/.style={transform shape}]
{
              & & \node [or gate,small circuit symbols] (nor) {\,\,\,OR}; & \\
\node (k) {}; & & 
\node [draw,rectangle,rounded corners,minimum width=10mm,minimum height=4.5mm,xshift=-0mm] (chan) {$c$}; & \\
};

\draw[<-] (nor.input 1)
        -- ++(left:1.4) node[left] (i) {$i$};
\draw[->] (nor.output) 
	-- ++(right:0.7) node[circle,inner sep=0.7pt,fill=black,draw] (huhu) {}
	-- ++(down:0.1)
        |- (chan.east);
\draw[->] (chan.west)
        -- ++(left:1.0) 
	-- ++(up:0.1)
	|- (nor.input 2)
	;
\draw[->] (huhu)
        -- ++(right:1.2)
node [draw,rectangle,fill=white,minimum width=7mm,minimum
height=5mm,xshift=-0mm] (ht) {HT}
	-- ++(right:1.2) node[right] {$o$}
	;
\end{tikzpicture}
\caption{A circuit solving unbounded SPF, consisting of an OR-gate fed back by
channel~$c$, and a high-threshold filter HT.}
\label{fig:circuit}
\end{figure}
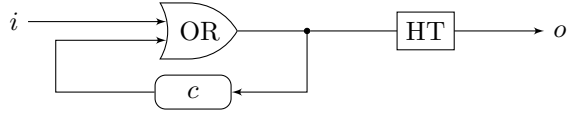

We consider a pulse of length~$\Delta$ at time~$0$ at the input and reason
about the behavior of the feed-back loop. Then, we show that this behavior 
can be translated to a legitimate SPF output by using a high-threshold filter.
We start by identifying two extremal cases: 
\ifthenelse{\boolean{SHORTversion}}{
If $\Delta$ is too small, i.e., $\Delta\leq \delta_\infty - \delta_{\min}$, then
the pulse is filtered by the channel in the feed-back loop. 
If it is too large, i.e., $\Delta\geq \delta_\infty$,
the pulse is captured by the storage loop, leading to a stable output~1.
}
{
If $\Delta$ is too small, then
the pulse is filtered by the channel in the feed-back loop. 
If it is too big,
the pulse is captured by the storage loop, leading to a stable output~1.
\begin{lemma}
\label{lem:big:pulse}
If the input pulse's length $\Delta$ satisfies $\Delta\geq\delta_\infty$, 
then the output of the OR has a unique rising transition at
time~$\delta_\infty$.
\end{lemma}
\begin{proof}
%
%
%
%
Assigning the channel output~$s_c$ a single rising transition at time~$\delta_\infty$ is
part of a consistent execution, in which the OR's output has a single rising
transition at time~$0$.
The lemma now follows from uniqueness of executions.
\end{proof}

\begin{lemma}
\label{lem:small:pulse}
If the input pulse's length $\Delta$ satisfies
$\Delta \leq \delta_\infty - \delta_{\min}$, then the OR output 
contains only the input pulse.
\end{lemma}
\begin{proof}
The input signal contains only two transitions: One at time $t_1=0$ and one at
time $t_2=\Delta\leq \delta_\infty - \delta_{\min}$.
Since $\delta_1 = \delta_\infty$ and hence 
$t_2 = t_1 + \Delta \leq t_1 + \delta_1 - \delta_{\min}$,
the two pending transitions of $c$'s output cancel by Lemma~\ref{lem:minimal:delay},
and no further transitions are generated afterwards.
\end{proof}
}

Now suppose that the input pulse length satisfies $\delta_\infty-\delta_{\min}<
\Delta_0< \delta_\infty$.
For these pulse lengths~$\Delta_0$, the OR output signal will contain a series of
pulses of lengths $\Delta_0,\Delta_1,\Delta_2,\dots$
For all but one~$\Delta_0$, this series will turn out to be either decreasing or
increasing and finite, causing the output signal to be eventually~$0$ or eventually~$1$.
To compute these pulse lengths, we define the auxiliary function
\begin{equation}
f(\Delta) = \delta\big( \Delta - \delta(-\Delta) \big) + \Delta -
\delta(-\Delta)\label{def:ffunc}
\enspace,
\end{equation}
which gives $\Delta_{n} = f(\Delta_{n-1})$ for all $n\geq2$.
To see this, note that $\Delta_{n-1}$ at the channel input
is also present at the channel output, so the rising resp.\ falling transition
is delayed by $\delta(-\Delta_{n-1})$ resp.\
$\delta( \Delta_{n-1} - \delta(-\Delta_{n-1}))$.
The first generated pulse starts from a zero channel input and thus
\ifthenelse{\boolean{SHORTversion}}{
$\Delta_1 = \Delta_0 - \delta_\infty + \delta(\Delta_0 - \delta_\infty)$.
}{
\begin{equation}
\Delta_1 = \Delta_0 - \delta_\infty + \delta(\Delta_0 - \delta_\infty).\label{eq:firstdelta}
\end{equation}
}

The procedure stops if either $f(\Delta_n)\leq0$ (pulse canceled; the output 
is constant~$0$ thereafter), or if 
\ifthenelse{\boolean{SHORTversion}}{
$f(\Delta_n)\geq\delta(0) > 0$
}{
\begin{equation}
f(\Delta_n)\geq\delta(0) > 0\label{eq:pulsecaptured}
\end{equation}
}
(pulse captured; the output is constant~$1$ thereafter).
%
%
%
%

The only case in which the procedure does not stop is if
$f(\Delta_1) = \Delta_1$.
There is a unique $\Delta_1>0$ with this property, denoted~$\tilde{\Delta}_1$.
By (\ref{def:ffunc}), it is also characterized by the relation $\delta(-\tilde{\Delta}_1) =
2\tilde{\Delta}_1$. Since $\delta(-\delta(0))=0$ by the involution property,
we must have $\tilde{\Delta}_1 < \delta(0)$.
Since $\Delta_1\to \delta(0)$ as $\Delta_0\to\delta_\infty$ and $\Delta_1\to
0$ as $\Delta_0\to\delta_\infty - \delta_{\min}$,
there exists a unique~$\Delta_0$ such that $\Delta_1 = \tilde{\Delta}_1$.
Denote it by~$\tilde{\Delta}_0$.

\ifthenelse{\boolean{SHORTversion}}{
This gives the following characterization of the feed-back loop's behavior:

}{
The following lemma shows that the procedure indeed stops if and only if
$\Delta_1\neq\tilde{\Delta}_1$, and can be used to bound the number of steps
until it stops.

\begin{lemma}\label{lem:exponential:decay}
$\lvert f(\Delta_1) - \tilde{\Delta}_1 \rvert \geq (1+\delta'(0)) \cdot \lvert
\Delta_1 - \tilde{\Delta}_1 \rvert$
if $\Delta_1>0$.
\end{lemma}
\begin{proof}
We have
\begin{equation}
\begin{split}
f'(\Delta_1)  = &
\big( 1 + \delta'(-\Delta_1) \big)\cdot \delta'\big(\Delta_1 -
\delta(-\Delta_1)\big)
\\ & + 1 +
\delta'(-\Delta_1)
\geq 1 + \delta'(0)
\end{split}
\end{equation}
because $\delta'(-\Delta_1)\geq \delta'(0)$  and $\delta'(T)>0$
for all~$T$ as $\delta$ is concave and increasing.
The mean value theorem of calculus now implies the lemma.
\end{proof}
}

\begin{theorem}\label{thm:or:loop}
The fed-back OR gate with a strictly causal symmetric involution 
channel has the following output when the
input pulse has length~$\Delta_0$:
\begin{itemize}
\item If $\Delta_0 > \tilde{\Delta}_0$, then the output is eventually
constant~$1$.
\item If $\Delta_0 < \tilde{\Delta}_0$, then the output is eventually
constant~$0$.
\item If $\Delta_0 = \tilde{\Delta}_0$, then the output is a periodic pulse
train with duty cycle 50\%.
\end{itemize}
Furthermore, the stabilization time in the first two cases is in the order of
$\log (1/\lvert \Delta_0 - \tilde{\Delta}_0 \rvert)$.
\end{theorem}
\ifthenelse{\boolean{SHORTversion}}{}{
\begin{proof}
If $\Delta_0\geq \delta_\infty$ or $\Delta_0\leq \delta_\infty -
\delta_{\min}$, then Lemmas~\ref{lem:big:pulse} and~\ref{lem:small:pulse}
show the theorem.

So let $\Delta_0\in (\delta_\infty - \delta_{\min} , \delta_\infty)$.
By Lemma~\ref{lem:exponential:decay}, the number of generated pulses until the
procedure stops is in the order of $\log 1 / \lvert \Delta_1 -
\tilde{\Delta}_1\rvert$.
Setting $g(\Delta_0) = \Delta_0 - \delta_\infty + \delta(\Delta_0 -
\delta_\infty)$ such that $\Delta_1=g(\Delta_0)$, cp.\ (\ref{eq:firstdelta}), 
and applying the mean value theorem of calculus to this
function, we see analogously as in the proof of
Lemma~\ref{lem:exponential:decay} that
\[
\lvert \Delta_1 - \tilde{\Delta}_1 \rvert
\geq 
\big( 1 + \delta'(0) \big) \cdot
\lvert \Delta_0 - \tilde{\Delta}_0 \rvert
\enspace.
\]
Hence the number of generated pulses is in the order of $\log 1 / \lvert
\Delta_0 - \tilde{\Delta}_0\rvert$.
Since both the length $\Delta_n$ of the occurring pulses 
and, by symmetry, the time between them is at most $\delta(0)$
since it would be captured otherwise, cp.\ (\ref{eq:pulsecaptured}),
we have the same asymptotic bound on the stabilization
time.
\end{proof}
}

\ifthenelse{\boolean{SHORTversion}}{

Finally, one can show that a high-threshold filter
with arbitrary threshold can be modeled
by an exp-channel with properly chosen~$V_{th}$:

}{
We now turn to the analysis of the high-threshold filter.

\begin{lemma}\label{lem:duty:cycle:delta}
Let~$c$ be an exp-channel~$c$ with threshold~$V_{th}$.
Then there exists some~$\Delta>0$ such that every periodic pulse train with pulse
lengths at most~$\Delta$ and duty cycle (ratio of 1-to-0) at most~$V_{th}$ is mapped to the zero
signal by~$c$.
\end{lemma}
\begin{proof}
Let $t_1,t_2,\dots$ be the times of transitions in the input pulse train with
duty cycle $\gamma\leq V_{th}$, i.e.,
$t_1=0$, $t_{2n+2} = t_{2n+1}+\Delta$, and $t_{2n+1} = t_{2n} + \Delta/\gamma$.
We assume that $\Delta$ is smaller than both $-\tau\log(1-V_{th})$ and a
to-be-determined $\Delta_0$, and show inductively that all pulses get canceled:

If $\Delta\leq -\tau\log(1-V_{th})$, then 
$\delta_1=\delta_\uparrow(\infty)=T_p-\tau\log(1-V_{th})>T_p$ and 
$\delta_2=\delta_\downarrow(\Delta-\delta_1) \leq T_p$, so the first pulse is canceled.
For the induction step, we assume $\delta_{2n}\leq T_p$ and find
\[
\begin{split}
\delta_{2n+1} & = \delta_\uparrow(\Delta/\gamma - \delta_{2n})
\geq \delta_\uparrow(\Delta/V_{th} - T_p)
\\ &
= T_p - \tau\log(1-V_{th}) + \tau\log(1-V_{th} e^{-\Delta/V_{th}\tau})
\end{split}
\]
Hence, $t_{2n+1}$ and $t_{2n+2}=t_{2n+1}+\Delta$ cancel if 
\[
\begin{split}
\Delta & \leq
\delta_{2n+1}-T_p
\\ &
=
-\tau\log(1-V_{th}) + \tau\log(1-V_{th} e^{-\Delta/V_{th}\tau})
\enspace,
\end{split}
\]
which is equivalent to 
\[
\begin{split}
h(\Delta)
=
V_{th} e^{-\Delta/V_{th}\tau}
+
(1-V_{th})e^{\Delta/\tau} \leq 1
\enspace.
\end{split}
\]
The latter satisfies $h(0)=1$ and 
\[
\begin{split}
h'(\Delta)
=
-\frac{1}{\tau} e^{-\Delta/V_{th}\tau}
+
\frac{1-V_{th}}{\tau}e^{\Delta/\tau}
\enspace,
\end{split}
\]
in particular $h'(0)=-V_{th}/\tau<0$.
There hence exists some~$\Delta_0>0$ such that $h(\Delta)\leq 1$ for all
$0\leq\Delta\leq\Delta_0$.
Thus, $t_{2n+1}$ and $t_{2n+2}$ cancel, and
\[
\begin{split}
\delta_{2n+2} & =
\delta(\Delta - \delta_{2n+1})
\leq
\delta(-T_p)=T_p
\end{split}
\]
because $h(\Delta)\leq1$, which completes the induction step.
\end{proof}

By letting~$\tau$ grow, one can even achieve the following result.

}

\begin{lemma}\label{lem:ht:exists}
Let~$\hat{\Delta}>0$ and $0<\gamma<1$.
Then there exists an exp-channel with threshold $V_{th}=\gamma$ such
that every periodic pulse train with pulse
lengths at most~$\hat{\Delta}$ and duty cycle at most~$\gamma$ is mapped to the zero
signal by~$c$.
\end{lemma}
\ifthenelse{\boolean{SHORTversion}}{
}{
\begin{proof}
We use the notation of the proof of Lemma~\ref{lem:duty:cycle:delta}.
The unique root of $h'(\Delta)$ is equal to
\[
\Delta_\tau=-\frac{\tau\log(1-V_{th})}{1+1/V_{th}}
\enspace,
\]
which goes to infinity as $\tau\to\infty$.
We can choose $\Delta_0=\Delta_\tau$ because $h'(\Delta)\leq0$ for all
$0\leq\Delta\leq\Delta_\tau$.
Because also $-\tau\log(1-V_{th})$ goes to infinity as $\tau\to\infty$, we can
find, for any given $\Delta$, some $\tau>0$ such that both $\Delta\leq
-\tau\log(1-V_{th})$ and $\Delta\leq \Delta_\tau$.
But for these $\Delta$, all input pulse trains with pulse lengths~$\Delta$ and
duty cycle at most $V_{th}=\gamma$ get mapped to the zero signal.
\end{proof}
}

In particular, by choosing $\gamma = 0.6$ and~$\hat{\Delta}$ large enough such that the
output of the feed-back loop is already constant~$1$ at time $T+\hat{\Delta}$ if the
duty cycle in the loop passes $0.6$ at time~$T$, the critical pulse duration 
$\tilde{\Delta}_0$ is mapped to a zero-output. It hence follows:

\begin{theorem}
There is a circuit that solves unbounded SPF.
\end{theorem}


%
%

\section{Impossibility of Bounded Short-Pulse
Filtration}\label{sec:impossibility}

\ifthenelse{\boolean{SHORTversion}}{

We first show that that strictly causal involution channels are
continuous in a certain sense that we will define precisely below.
We start with a suitable distance of signals.

\begin{definition}
For a signal~$s$ and a time~$T$, denote by $\mu_T(s)$ the total duration in
$[0,T]$ where~$s$ is~$1$.
That is, $\mu_T(s)$ is the measure of the set $\{t\in[0,T]\mid
s(t)=1\}$.

For any two signals~$s_1$ and~$s_2$ and every~$T$, we define their {\em
distance up to time~$T$\/} by setting $\lVert s_1-s_2\rVert_T = \mu_T(\lvert
s_1-s_2\rvert)$.
\end{definition}

Intuitively, an involution channel is continuous under this measure
for two reasons: (i) Due to the continuity of~$\delta$, a small change in
the time at which an input transition occurs, results in a small change
in the time at which the corresponding output transition occurs.
This, again, only results in a small change of the input-to-previous-output
time for the next input transition, and so on. The technical difficulty
is to show that this effect does not result in discontinuities even for an
unbounded number of
input transitions. (ii) Due to the involution property of~$\delta$,
one can show that~$\delta$ is not only continuous in changing the length of input
pulses, but also in removing them: An input pulse whose length
is arbitrarily small results in a value of~$\delta$ for the next input transition
that is arbitrarily close to the transition's $\delta$ value in the case the short pulse
was not present at all. Again, the major difficulty lies in showing that this also
holds for infinite pulse trains. 
Note carefully that it is primarily the continuity property~(ii) that
distinguishes our involution channels from the ``unfaithful'' single-history channels 
analyzed in \cite{FNS13}, which allow bounded SPF to be solved. 

The detailed proof in Appendix~\ref{sec:contproof} establishes:

\begin{theorem}
\label{thm:channels:are:continuous}
Let~$c$ be a channel and let~$T\in [0,\infty)$.
Then, the mapping $s\mapsto c(s)$ is continuous with respect to the distance
$\lVert s_1 - s_2 \rVert_T$.
\end{theorem}

}{
\subsection{Continuity of Channels}

In this subsection, we prove that strictly causal channels are
continuous in a certain sense that we will define precisely.
For ease of exposition and for space reasons, we give the proof only in the
case of symmetric channels, i.e., for the case that
$\delta_\uparrow=\delta_\downarrow=\delta$.

We begin by noting that channels are monotone.
To compare certain signals, we write $s_1\leq s_2$ if~$s_2$ is~$1$ whenever~$s_1$
is.

\begin{lemma}
\label{lem:channels:are:increasing}
Let~$s_1$ and~$s_2$ be signals such that $s_1\leq s_2$ and let~$c$ be a
channel.
Then, $c(s_1)\leq c(s_2)$.
\end{lemma}

We next define a distance for signals, for which channels will turn out to be
continuous.

\begin{definition}
For a signal~$s$ and a time~$T$, denote by $\mu_T(s)$ the total duration in
$[0,T]$ where~$s$ is~$1$.
In more symbolic terms, $\mu_T(s)$ is the measure of the set $\{t\in[0,T]\mid
s(t)=1\}$.

For any two signals~$s_1$ and~$s_2$ and every~$T$, we define their {\em
distance up to time~$T$\/} by setting $\lVert s_1-s_2\rVert_T = \mu_T(\lvert
s_1-s_2\rvert)$.
\end{definition}

Intuitively, an involution channel is continuous under this measure
for two reasons: (i) Due to the continuity of~$\delta$, a small change in
the time at which an input transition occurs, results in a small change
in the time at which the corresponding output transition occurs.
This, again, only results in a small change of the input-to-previous-output
time for the next input transition, and so on. The technical difficulty
is to show that this effect does not result in discontinuities even for an
unbounded number of
input transitions. (ii) Due to the involution property of~$\delta$,
one can show that~$\delta$ is not only continuous in changing the length of input
pulses, but also in removing them: An input pulse whose length
is arbitrary small results in a value of~$\delta$ for the next input transition
that is arbitrarily close to the transition's $\delta$ value in the case the short pulse
was not present at all. Again, the major difficulty lies in showing that this also
holds for infinite pulse trains. 

Note carefully that it is primarily the continuity property~(ii) that
distinguishes our involution channels from the ``unfaithful''
single-history channels 
analyzed in \cite{FNS13}, which allow bounded SPF to be solved. 

\medskip

We start our detailed proof with Lemma~\ref{lem:add:pulse:to:the:end}, 
which provides an optimal choice for appending a pulse at the
end of a signal in order to maximize $\mu_T$.
We will use it later when bounding the maximum impact of an infinitesimally small
pulse.

We use the shorthand notation $(x)_+$ for $\max(x,0)$.

\begin{lemma}
\label{lem:add:pulse:to:the:end}
Let~$s$ be a signal that is eventually constant~$0$ and let~$c$ be a channel.
Denote by~$t_n$ the time of the last (falling) transition in~$s$ and
by~$\delta_n$ its
delay in the channel algorithm for~$c$. 
Then, the maximal $\mu_T(c(s'))$ among
all~$s'$ obtained from~$s$ by appending one pulse of length~$\Delta$ after
time~$t_n$ is attained by the addition of the pulse at time $t_n +(\delta_n -
\delta_{\min})_+$ (which results in a cancellation, i.e., a right-shift, of the 
last transition).
\end{lemma}
\begin{proof}
We first show the lemma for $T=\infty$ and then extend the result to
finite~$T$.
Let~$s'_\gamma$ be the addition of the pulse of length~$\Delta$ to~$s$ at time
$t_n + \gamma$.

For all $0\leq\gamma\leq\delta_n-\delta_{\min}$, set
\[
f(\gamma) = \gamma + \delta(\Delta - \delta(\gamma - \delta_n))
\enspace.
\]
In the class of all~$s'_\gamma$ with $0\leq\gamma\leq \delta_n-\delta_{\min}$
(which can be empty), the
maximum of~$\mu_T(c(s'_\gamma))$ is attained at the maximum of~$f$.
This is because the transition at time $t_n+\gamma$ cancels that at time~$t_n$
in this case.
The derivative of~$f$ is equal to
\[
f'(\gamma) = 1 - \delta'(\Delta - \delta(\gamma - \delta_n))\cdot
\delta'(\gamma - \delta_n)
\enspace.
\]
The condition $f'(\gamma) = 0$ is equivalent to $\delta'(\Delta - \delta(\gamma
- \delta_n)) = 1 / \delta'(\gamma - \delta_n)$, which is in turn equivalent to
  $\delta(\Delta - \delta(\gamma - \delta_n)) = -(\gamma - \delta_n)$, i.e.,
$\Delta=0$, as $\delta'(t)=1/\delta'(-\delta(t))$ by Lemma~\ref{lem:delta:min}.
Hence, $f'(\gamma)$ is never zero.
Since $f'(\gamma)\to 1$ as $\gamma\to\infty$, as the concave $\delta$ satisfies
$\lim_{t\to\infty}\delta'(t)=0$, the derivative of~$f$ is always
positive, hence~$f$ is increasing.
This shows that $\gamma = \delta_n-\delta_{\min}$ is a strictly better choice
than any other $\gamma$ in this class.

For the class of~$s'_\gamma$ with $\gamma \geq (\delta_n - \delta_{\min})_+>0$, we
define the function
\[
g(\gamma) = \Delta+ \delta(\Delta - \delta(\gamma - \delta_n)) -
\delta(\gamma - \delta_n)
\enspace.
\]
Since the transitions at~$t_n$ and $t_n+\gamma$ do not cancel in this class,
the maximum of $\mu_\infty(c(s'_\gamma))$ is attained at the maximum of~$g$.
But it easy to see, using the monotonicity of~$\delta$, that~$g$ is decreasing.
The maximum of~$g$ is hence attained at $\gamma = (\delta_n- \delta_{\min})_+$. 

Consequently, the choice $\gamma = \gamma_0 = (\delta_n- \delta_{\min})_+$
maximizes $\mu_\infty(c(s'_\gamma))$ in any case. By Lemma~\ref{lem:minimal:delay}, this choice
results in a cancellation of the last (falling) transition in $s$, hence
a right-shift of the latter in $s'$. This concludes our proof for $T=\infty$.

Let now~$T$ be finite.
Denote by~$T_0$ the time of the last, falling, output transition in
$c(s_{\gamma_0}')$.
In this case, transitions of $c(s)$ and $c(s_{\gamma_0}')$ are the same except
the last,
falling, transition which is delayed from $t_n + \delta_n$ to~$T_0$.
We distinguish the two cases (a) $T\leq T_0$ and (b) $T > T_0$.
In case~(a), the last transition of $c(s)$ is delayed beyond~$T$ in
$c(s_{\gamma_0}')$.
Because all other transitions remain unchanged in all $c(s_\gamma')$, the
measure $\mu_T(c(s_{\gamma_0}'))$ is maximal among all $\mu_T(c(s_\gamma'))$
if $T\leq T_0$.
In case~(b), we have $\mu_T(c(s_{\gamma_0}')) = \mu_\infty(c(s_{\gamma_0}'))$.
But because $\mu_T \leq \mu_\infty$ and $\mu_\infty(c(s_{\gamma_0}'))$ is
maximal among all $\mu_\infty(c(s_{\gamma}'))$, so is
$\mu_T(c(s_{\gamma_0}'))$ among all $\mu_T(c(s_{\gamma}'))$.
\end{proof}

We next effectively bound the maximum impact on~$\mu_T$ that a set of 
pulses of small combined length can have.

\begin{lemma}
\label{lem:epsilon:pulses:at:the:end}
Let~$s$ be a signal that is eventually constant~$0$ and let~$c$ be a channel.
Then there exists a constant~$d$ such that
the maximal~$\mu_T(c(s'))$ among all~$s'$ obtained from~$s$ by adding
pulses of combined length~$\varepsilon$ after the last transition of~$s$ is at
most $\mu_T(c(s)) + d\cdot\varepsilon$.
\end{lemma}
\begin{proof}
It suffices to show the lemma for $T=\infty$.
Let~$\varepsilon=\sum_{k=1}^\infty \varepsilon_k$.
We add, one after the other, pulses of length~$\varepsilon_k$ after the last
transition.
We show that the maximum gain after adding~$K$ pulses is at most $\sum_{k=1}^K
\varepsilon_k$.

Denote by~$t_n$ the last transition in~$s$ and by~$\delta_n$ its delay.
By Lemma~\ref{lem:add:pulse:to:the:end}, it is optimal to add the first pulse
(of length~$\varepsilon_1$) at time $t_n + (\delta_n-\delta_{\min})_+$;
call the resulting signal~$s_1'$.

We first assume $\delta_n-\delta_{\min}\geq0$.
Here, the two new transitions in~$s_1'$ are $t_{n+1} = t_n + \delta_n-\delta_{\min}$ and
$t_{n+2} = t_n + \delta_n-\delta_{\min} + \varepsilon_1$.
Their corresponding delays are $\delta_{n+1}=\delta_{\min}$ and $\delta_{n+2} =
\delta(\varepsilon_1 - \delta_{\min})$.
By the mean value theorem of calculus and Lemma~\ref{lem:delta:min}, the
duration of the resulting pulse is
\[
\begin{split}
\delta_{n+2} - \delta_{n+1}  = 
\delta(\varepsilon_1 -\delta_{\min}) - \delta(- \delta_{\min})
 = \varepsilon_1\cdot \delta'(\xi)
\end{split}
\]
for some $-\delta_{\min} \leq \xi \leq \varepsilon_1 - \delta_{\min}$.
Since~$\delta'$ is decreasing and $\delta'(-\delta_{\min})=1$
according to Lemma~\ref{lem:delta:min}, we hence deduce
$0\leq\delta_{n+2} - \delta_{n+1} \leq \varepsilon_1$.
Thus $\mu_T(c(s_1') - c(s)) = \varepsilon_1 + \delta_{n+2} - \delta_{n+1} \leq
2\varepsilon_1$.
Since $\delta_{n+2} > \delta_{\min}$, we can 
continue this argument inductively.

If now $\delta_n - \delta_{\min}<0$, then~$t_n$ is effectively replaced  by
$t_n+\varepsilon_1$ in~$s_1'$, i.e., right-shifted.
This changes the measure by
\[
\big(
\delta(t_n - t_{n-1} + \varepsilon_1 - \delta_{n-1})
-
\delta(t_n - t_{n-1} - \delta_{n-1})
\big)_+
\enspace,
\]
which is at most $\varepsilon_1 \cdot \delta'(t_n - t_{n-1} - \delta_{n-1})$ by
the mean value theorem.
We note that this second case only occurs until the first case
happens one time.
We can hence merge all the~$\varepsilon_k$ of the first case and set $d =
\max(2, \delta'(t_n - t_{n-1} - \delta_{n-1}))$.
\end{proof}

We combine the previous two lemmas to show continuity of channels:

\begin{theorem}
\label{thm:channels:are:continuous}
Let~$c$ be a channel and let~$T\in [0,\infty)$.
Then, the mapping $s\mapsto c(s)$ is continuous with respect to the distance
$d(s_1,s_2) = \lVert s_1 - s_2 \rVert_T$.
\end{theorem}
\begin{proof}
Let~$s$ be a signal.
We show that, if $\lVert s - s_n \rVert_T \to 0$, then $\lVert c(s) - c(s_n)
\rVert_T \to 0$.
Because
\[
\lvert s - s_n\rvert = (\max(s,s_n) - s) + (s - \min(s,s_n))
\enspace,
\]
where $\max(s,s_n)(t)=\max(s(t),s_n(t))$ and $\min(s,s_n)(t)=\min(s(t),s_n(t))$ for all $t$,
the condition $\lVert s-s_n\rVert_T\to 0$ is equivalent to conjunction of $\lVert s
- \max(s,s_n)\rVert_T\to0$ and $\lVert s - \min(s,s_n)\rVert_T\to0$.
Because
$\max(c(s),c(s_n)) \leq c(\max(s,s_n))$
and
$\min(c(s),c(s_n)) \geq c(\min(s,s_n))$
by Lemma~\ref{lem:channels:are:increasing},
we have
\[
\begin{split}
\lvert c(s) - c(s_n) \rvert
\leq
\ &
c(\max(s,s_n)) - c(s)
\\&+ c(s) - c(\min(s,s_n))
\enspace,
\end{split}
\]
which shows that
we can
suppose without loss of generality $s_n\geq s$ for all~$n$.

Let $(t_m,0),(t_{m+1},1)$ be a negative pulse in~$s$.
Since there are only finitely many negative pulses before time~$T$, it suffices
to show $\mu_T(c(s_n)-c(s))\to0$ in the case that $s_n-s$ is zero
outside of $[t_m,t_{m+1}]$, i.e., that the only additions of~$s_n$ with respect
to~$s$ lie in the given negative pulse.

Let~$\mu_T(s_n-s) \leq \varepsilon$.
It follows from Lemma~\ref{lem:epsilon:pulses:at:the:end} that the increase in
measure incurred directly from the new pulses is $O(\varepsilon)$.
Furthermore, by Lemma~\ref{lem:add:pulse:to:the:end}, the measure incurred by
later transitions~$t_k$ with~$k>m$
are biggest when merging all new pulses at the end of the negative pulse.
Because the delays of these transitions
depend continuously on~$\varepsilon$ and~$\mu_T(c(s_n)-c(s))$ depends
continuously on these delays, we have $\mu_T(c(s_n)-c(s))\to 0$ as
$\varepsilon\to0$.
\end{proof}

\subsection{Impossibility in Forward Circuits}
}

Call a circuit a {\em forward circuit\/} if its graph is acyclic.
Forward circuits are exactly those circuits that do not contain
     feed-back loops.
Equipped with the continuity of involution channels and the fact that
     the composition of continuous functions is continuous, it is not
     too difficult to prove that the inherently discontinuous SPF problem
     cannot be solved with forward circuits.

\begin{theorem}\label{thm:no_forward_circuit}
No forward circuit solves bounded SPF.
\end{theorem}
\ifthenelse{\boolean{SHORTversion}}{
}{
\begin{proof}
Suppose that there exists a forward circuit that solves bounded SPF with
stabilization time bound~$K$.
Denote by~$s_\Delta$ its output signal when feeding it a $\Delta$-pulse at
time~$0$ as the input.
Because~$s_\Delta$ in forward circuits is a finite composition of continuous
functions by
Theorem~\ref{thm:channels:are:continuous}, the measure $\mu_T(s_\Delta)$
depends continuously on~$\Delta$.

By the nontriviality condition (F3) of the SPF problem, there exists
some~$\Delta_0$ such that $s_{\Delta_0}$ is not zero.
Set $T = 2\Delta_0 + K$.

Let~$\varepsilon>0$ be smaller than $\mu_T(s_{\Delta_0})$.
We show a contradiction by finding a~$\Delta$ such that~$s_\Delta$ either
contains a pulse of length less than~$\varepsilon$ (contradiction to the no
short pulses condition (F4)) or contains a transition
after time $\Delta+K$ (contradicting the bounded stabilization time
condition~(F5)).

Since $\mu_T(s_\Delta)\to0$ as $\Delta\to0$ by the no generation condition (F2)
of SPF,
there exists a~$\Delta_1<\Delta_0$ such that $\mu_T(s_{\Delta_1})=\varepsilon$
by the intermediate value property of continuity.
By the bounded stabilization time condition (F5), there are no transitions
in~$s_{\Delta_1}$ after time $\Delta_1+K$.
Hence, $s_{\Delta_1}$ is~$0$ after this time because otherwise it is~$1$ for the
remaining duration $T - (\Delta_1+K) > \Delta_0 > \varepsilon$, which would
mean that $\mu_T(s_{\Delta_1})>\varepsilon$.
Consequently, there exists a pulse in~$s_{\Delta_1}$ before
time $\Delta_1+K$.
But any such pulse is of length at most~$\varepsilon$ because
$\mu_{\Delta_1+K}(s_{\Delta_1}) \leq \mu_T(s_{\Delta_1})=\varepsilon$.
This is a contradiction to the no short pulses condition (F4).
\end{proof}

\subsection{Simulation with Unrolled Circuits}
}

We next show how to simulate (part of) an execution of an arbitrary
     circuit~$C$ by a forward circuit~$C'$ generated from~$C$ by
     unrolling of feedback channels.
Intuitively, the deeper the unrolling, the longer the time~$C'$
     behaves as~$C$.

\begin{definition}
Let~$C$ be a circuit with input~$i$.
For~$v$ being a gate or input in~$C$ and $k\ge 0$, the {\em $k$-unrolled
     circuit $C_k(v)$} is
     constructed inductively as follows: If~$v=i$, or $v$ is a gate
     with no predecessor in~$C$, then~$C_k(v)$ is the
     circuit 
     that comprises only
     of vertex~$v$ and whose output is~$v$. (We slightly misuse the circuit
     definition here by allowing circuits consisting of a single vertex only.)
Otherwise, $v$ is a gate with predecessors and we distinguish two
     cases:

If~$k=0$, $C_k(v)$ comprises of gate~$v^{(\alpha)}$, with~$\alpha$
     being a unique identifier, and for each predecessor~$\sigma$
     of~$v$ in~$C$: if~$\sigma=i$, add~$i$ and an edge from~$i$
     to~$v^{(\alpha)}$; if~$\sigma$ is a channel, add
     channel~$\sigma^{(\beta)}$ and gate~$\tilde{x}^{(\gamma)}$,
     with~$\beta$ and~$\gamma$ being unique identifiers and~$x$ being
     the channel's initial value.
The Boolean function assigned to~$\tilde{x}^{(\gamma)}$ is
     constant~$x$.
The channel functions of~$\sigma^{(\beta)}$ and~$\sigma$ are the same.
Furthermore, add edges from $\tilde{x}^{(\gamma)}$ to~$\sigma^{(\beta)}$
     and from~$\sigma^{(\beta)}$ to~$v^{(\alpha)}$.
The Boolean function assigned to~$v^{(\alpha)}$ is the same as for~$v$
     and the ordering of the predecessors of~$v^{(\alpha)}$ reflects
     the ordering of the predecessors of~$v$.

If~$k > 0$, $C_k(v)$ is the circuit that comprises of
     gate~$v^{(\alpha)}$, with unique identifier~$\alpha$, and for
     each predecessor~$\sigma$ of~$v$ in circuit~$C$:
     If~$\sigma$ is a channel, let~$w$ be its predecessor in~$C$.
Add and connect the output of circuit~$C_{k-1}(w)$ to a
     channel~$\sigma^{(\beta)}$ and the channel to~$v^{(\alpha)}$.
If~$\sigma=i$, add~$i$ and connect it to~$v^{(\alpha)}$.
Again, the Boolean functions, orderings and channel functions are
     assigned in accordance with those in~$C$. 

In all cases, we say that the vertices~$\sigma^{(\alpha)}$ {\em
	correspond
     to~$\sigma$}.
\end{definition}

Let~$o$ be the single output of circuit~$C$.
To each vertex~$\sigma$ in~$C_k(o)$, we assign a value~$z(\sigma)$
     from~$\IN_0\cup\{\infty\}$ as follows: $z(\tilde{0}^{(\alpha)})=
     z(\tilde{1}^{(\alpha)}) = 0$, $z(i) = z(\sigma) = \infty$
     if~$\sigma$ has no predecessor in~$C$, $z(\sigma) = 1+z(w)$ for a
     channel~$\sigma$ with predecessor~$w$, and $z(\sigma) =
     \min\{z(\sigma') \mid \sigma' \text{ is a predecessor of }
     \sigma\}$ for a gate~$\sigma$.
\figref{fig:unrolling} shows an example of a circuit and an
     unrolled circuit with $z$~values.

\begin{figure}
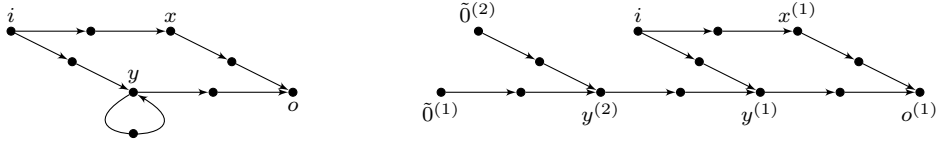

\centering
  {
    \tikzfigureunrolling}
\caption{Circuit~$C$ (left) and $C_2(o)$ (right) under the assumption that both incoming channels to gate~$y$ have initial value~$0$. It is $z(\tilde{0}^{(1)})=z(\tilde{0}^{(2)})=0$,
$z(i)=z(x^{(1)})=\infty$, $z(y^{(2)})=1$, $z(y^{(1)})=2$, and
$z(o^{(1)})=3$.}\label{fig:unrolling}
\end{figure}

We further adapt the execution construction algorithm in
     Section~\ref{sec:algorithm} to assign to each generated
     transition $e$ a {\em causal depth~$d(e)$}.
All initial transitions and input transitions have causal depth~$0$;
     all ``reset'' transitions initially added at time~$0$ have causal
     depth~$1$.
Algorithm step~(ii) is extended such that each transition~$e$ at
     time~$t$ that was marked fixed in~$s_v$, with~$v$ being a
     gate, gets assigned $d(e)$ equal to the maximum over all $d(e')$,
     where~$e'$ is a fixed transition at time~$t' \le t$ in
     $s_{\sigma}$ with $\sigma$ being a predecessor of~$v$.
Algorithm step~(iii) is extended such that for each transition~$e$
     in~$s_v$, with~$v$ being a gate or an input, that generates a
     transition~$e'$ in some~$s_{(v,w)}$, $d(e')=d(e)+1$.
We immediately get:   

\begin{lemma}\label{lem:Depth_Iteration}
For all~$k\ge 1$, (a) the constructive algorithm never assigns a
     causal depth larger than~$k$ to a transition marked fixed in
     iteration~$k$, and (b) at the end of iteration~$k$ the sequence
     of the causal depths of the transitions in~$s_{\sigma}$ is
     nondecreasing, for all vertices~$\sigma$.
\end{lemma}

\ifthenelse{\boolean{SHORTversion}}{

By an inductive argument on the causal depth of events generated in
both the circuit and its unrolled counterpart, one can show:

}{
We are now in the position to prove the main result of a circuit
     simulated by an unrolled circuit.
}

\begin{theorem}\label{thm:simulation}
Let~$C$ be a circuit with output port~$o$ that solves bounded SPF.
Let~$C_k(o)$ be an unrolling of~$C$, $\sigma$ a vertex in~$C$ and
     $\sigma'$ a vertex in~$C_k(o)$ corresponding to~$\sigma$.
For all input signals~$i$, if a transition~$e$ within~$s_\sigma$ is
     marked fixed by the execution construction algorithm run on
     circuit~$C$ with input signal~$i$
     and~$d(e) \le z(\sigma')$, then~$e$ is added and marked fixed
     in~$s_{\sigma'}$ by the algorithm run on circuit~$C_k(o)$
     with input signal~$i$; and vice versa.
\end{theorem}
\ifthenelse{\boolean{SHORTversion}}{

}{
\begin{proof}
We will show the statement by induction on~$d(e) \ge 0$ for the case
     where~$e$ is a transition in~$C$'s execution.
The proof for the case where~$e$ is a transition in~$C_k(o)$'s
     execution is analogous.

{\em Induction base:\/} By the construction of an unrolled circuit, $s_{\sigma}$
     and~$s_{\sigma'}$ have the same initial transitions if~$\sigma$
     is a gate or channel, and the same transitions
     if~$\sigma=\sigma'$ is the input.
Since, for all these transitions~$e$, $d(e) = 0 \le z(\sigma')$, the
     statement holds for~$d(e) = 0$.

{\em Induction step:\/} Assume that the lemma holds for all
     transitions~$e$ with~$d(e) \le k$.
We show that it also holds for transitions~$e'$ with $d(e')=k+1$
     if~$k+1\le z(\sigma')$.
If~$e'$ is a transition initially added by the constructive algorithm
     at time~$0$, $d(e')=1$.
From the definition of the unrolling, we immediately obtain that~$e'$
     is also added to~$s_{\sigma'}$ if~$z(\sigma') \ge 1$.
Otherwise, $e'$ must have been added within an iteration of the
     constructive algorithm.
Assume by contradiction that~$e'$ is the first transition (in the
     order transitions are generated by the constructive algorithm)
     with causal depth~$k+1$ added to a signal in~$C$ but not added
     to the respective signal in~$C_k(o)$.
We distinguish two cases for~$\sigma$: 

If~$\sigma=(v,w)$ is a channel in~$C$: Transition~$e'$ may only have
     been added to~$s_{(v,w)}$ by the channel algorithm with input
     signal~$s_v$ and a current transition~$e''$ with $d(e'')
     = k$.
The time of transition~$e'$ depends on the time of~$e''$ and on the
     last output transition only.
From the fact that~$e'$ is the first with depth~$k+1$ not added
     to~$s_{\sigma'}$, and the induction hypothesis applied to
     both~$s_{v}$ and~$s_{(v,w)}$, we deduce that transition~$e'$ is
     also added to~$s_{\sigma'}$; a contradiction to the assumption
     that it is the first one not added.

Since~$\sigma$ cannot be the input port, because all its transitions
     have causal depth~$0$ only, the case of~$\sigma$ being a gate in~$C$
     remains: However, then~$e'$ is generated due to a transition~$e''$ on
     a predecessor~$w$ of~$\sigma$ in~$C$.
Further, $d(e'') \le k+1$ and~$z(w') \ge k+1$ must hold for
     vertex~$w'$ corresponding to~$w$.
Since, $e'$ is the first transition with causal depth less or
     equal~$k+1$ not added to both circuits' signal, $e''$ was added to
     both~$s_{w'}$ and~$s_{w'}$, and thus~$e'$ is added to~$s_{\sigma'}$; 
     a contradiction that it is the first one not added.

This completes the induction step.
\end{proof}
}

\ifthenelse{\boolean{SHORTversion}}{

To finally deduce the impossibility of bounded SPF, we use the fact
     that a circuit $C$ that solves SPF with stabilization time
     bound~$K$ can be simulated by an unrolled circuit $C_N$ with
     depth~$N$ larger than the maximum causal depth of any of its
     transitions, i.e., larger than~$K/\delta_{\min}^C$ plus the number
     of input transitions; cf.\ Lemmas~\ref{lem:t_inc}
     and~\ref{lem:Depth_Iteration}.
This, however, contradicts the fact that no forward circuit, and thus
     specifically $C_N$, can solve bounded SPF.
We hence obtain our main result: 




}{
\subsection{The Impossibility Result}

Let the {\em aligned bounded SPF problem\/} be the SPF problem with
     the following modifications: We first require that if the input
     signal is a pulse, then the pulse must start at time~$0$ and be
     of length at most~$1$.
We call such signals {\em valid\/} input signals.
Further, we require that if the output signal~$o$ makes a transition
     to~$1$, it must do so before time~$K+1$, and~$o$ must remain~$1$
     from thereon until time~$K+2$, from whereon it is~$0$ until
     time~$K+3$ followed by a pulse of length~$1$ at time~$K+3$.
If the input is constant~$0$, we require that the output is a pulse of
     length~$1$ at time~$K+3$.
From every circuit that solves the (original) bounded SPF problem, we
     can easily build a circuit that solves the aligned SPF
     problem by adding capturing circuitry like Fig.~\ref{fig:circuit}.
In the following, we show that no circuit solves the aligned version of
     bounded SPF and thus, by the above reduction, the original
     bounded SPF problem.

Let~$C$ be a circuit that solves the aligned bounded SPF problem.
Then, for all input signals, the output signal~$o$ of~$C$ always
     contains a transition~$(K+4,0)$, regardless of the input.
Let~$D^C_o(i)$ be the causal depth of this transition in circuit~$C$
     when the input signal is~$i$.

\begin{lemma}\label{lem:clockbound}
Let~$C$ be a circuit that solves the aligned bounded SPF problem.
Then there exists an input signal~$i$ such that $D^C_o(i) >
     (K+5)/\delta_{\min}+2$.
\end{lemma}
\begin{proof}
Let $N = (K+5)/\delta_{\min}+2$, and assume for a contradiction that
$D^C_o(i) \le N$ for all valid input signals~$i$.
Consider the $N$-unrolled circuit~$C_N(o)$.
From Theorem~\ref{thm:simulation}, we obtain that transition~$(K+4,0)$,
     with causal depth in~$C$ at most~$N$ occurs at output~$o$ of~$C$
     if and only if it occurs at~$C_N(o)$'s output~$o'$ corresponding
     to~$o$.
Recalling Lemma~\ref{lem:Depth_Iteration}.(b), we obtain that the same holds
     for all transitions at output~$o$ with times less than~$K+4$;
     i.e., $C$'s and $C_N(o)$'s output signals restricted to
     time~$[-\infty,K+4]$ are the same for all valid input
     signals~$i$.
One can easily extend the forward circuit~$C_N(o)$ such that it remains a
     forward circuit and solves aligned bounded SPF, by suppressing all
     transitions at the output that possibly occur after time~$K+4$.
Since Theorem~\ref{thm:no_forward_circuit} also holds for the aligned
     bounded SPF problem, no such forward circuit exists; a contradiction
     to the initial assumption.
\end{proof}

From Lemma~\ref{lem:clockbound} and~\ref{lem:Depth_Iteration}, we
     obtain that, for input~$i$, the constructive algorithm does not
     mark fixed the output transition~$(K+4,0)$ before
     iteration~$(K+5)/\delta_{\min}+2$.
However, from Lemma~\ref{lem:t_inc} and the fact that input signal~$i$
     contains at most~$2$ transitions besides the initial transition
     at~$-\infty$, we conclude that all iterations~$\ell \ge
     (K+5)/\delta_{\min}+2$ have~$t_\ell \ge K+5$.
From Lemma~\ref{lem:t_fixed}, we conclude that all transitions still
     existent at the end of these iterations must have times at
     least~$K+5$; a contradiction to the fact that the final output transition
     occurs at time~$K+4$.
We thus obtain:                
}

\begin{theorem}\label{thm:main_impossibility}
No circuit solves bounded SPF.
\end{theorem}

\section{Conclusions and Future Work}

We showed that binary circuit models based on involution channels
are a promising candidate for faithfully modeling glitch propagation in real circuits,
in the sense that they allow to design circuits solving the 
Short-Pulse Filtration (SPF) problem precisely when this is
possible with physical circuits. Involution channels differ
from all existing single-history channels, which do not share
this property, in that they are also continuous with respect to dropping 
small input pulses.

Although our results prove that involution channels are superior 
to all alternative channel models known so far, there are several 
very important questions which are still open: 
First, we did not at all address the question
of quantitatively comparing the modeling accuracy of alternative
models: Although we believe that the modeling accuracy of properly chosen
instances of our involution channels surpasses the one of alternative channel
models, we cannot rule out the possibility that a non-faithful
model like PID works better in some situations. Second, 
addressing the SPF problem is only a first step
towards a digital model for metastability generation and
propagation. Needless to say, addressing both questions
requires major efforts and is hence a subject of future research.

%
%
%

\bibliographystyle{plain}
\bibliography{mybib}

\ifthenelse{\boolean{SHORTversion}}{
\appendix

\section{Proofs}
\label{sec:proofs}

\subsection{Proof of Lemma~\ref{lem:delta:min}}
Set $f(T) = -T + \delta_\uparrow(-T)$.
This function is continuous and strictly decreasing, 
since~$\delta_\uparrow$ is continuous and strictly increasing.
Because $f(0)=\delta_\uparrow(0)$ is positive and the limit of
$f(T)$ as
$T\to
\delta^\downarrow_\infty$ is $-\infty$, there exists a unique~$\delta_{\min}$
between~$0$
and~$\delta^\downarrow_\infty$ for which $f(\delta_{\min}) = 0$.
Hence, $\delta_\uparrow(-\delta_{\min})=\delta_{\min}$.
The second equality follows from 
$\delta_{\min}=\delta_\downarrow(-\delta_\uparrow(-\delta_{\min}))=\delta_\downarrow(-\delta_{\min})$
according to (\ref{eq:involution}).

The second part of the lemma follows by differentiating
Equation~\eqref{eq:involution}.

\subsection{Proof of Lemma~\ref{lem:minimal:delay}}

Let~$\delta$ be either~$\delta_\uparrow$ or~$\delta_\downarrow$, depending on
whether $t_{n+1}$ is a rising or falling transition.
By definition, the two transitions cancel if and only if
\begin{equation}\label{eq:lem:minimal:delay:proof}
\delta_{n+1} = \delta(t_{n+1}-t_n-\delta_n)
\leq
-(t_{n+1} - t_n - \delta_n)
\enspace.
\end{equation}
Set $T = t_{n+1} - t_n - \delta_n$.
By Lemma~\ref{lem:delta:min}, equality holds
in~\eqref{eq:lem:minimal:delay:proof} if and only if $T = -\delta_{\min}$.
Because the left-hand side of~\eqref{eq:lem:minimal:delay:proof} is increasing
in~$T$ and the right-hand side is strictly decreasing in~$T$,
\eqref{eq:lem:minimal:delay:proof} is equivalent to $T\leq -\delta_{\min}$,
which in turn is equivalent to $t_{n+1}\leq t_n + \delta_n - \delta_{\min}$.

\subsection{Proof of Lemma~\ref{lem:t_fixed}}

Statement~(a) is implied by the fact that transitions are only marked fixed in
     step~(i) and~(ii), which act on transitions at time~$t_\ell$ only.

For~(b), assume by contradiction that a
     transition $(t,X)$ with $t \leq
t_\ell + \delta_{\min}^C$ but different from~$t_\ell$ was added in
     iteration~$\ell$ and still exists at the end of iteration~$\ell$.
Such a transition can only be added via step~(iii).
For the respective channel algorithm with delay function~$\delta$ and minimal
delay~$\delta_{\min}$,
     $\delta(t_\ell-t') \leq \delta_{\min}^C\leq\delta_{\min}$ must have held, where~$t'$ is
     the time of the channel's last output transition.
However, by Lemma~\ref{lem:minimal:delay},
this leads to a cancellation and hence removal of $(t,X)$, which provides the required 
contradiction.

For~(c), assume by contradiction that, at the end of
     iteration~$\ell$, there exists a non-fixed transition $(t_\ell,X)$.
Since step~(i) marks all transitions at time~$t_\ell$ fixed and~(ii) adds
     only fixed transitions at time~$t_\ell$, the non-fixed transition must have
     been newly added in step~(iii).
However, as in~(b), we know that this requires $\delta(t_\ell-t') \leq
\delta_{\min}$ leading to a cancellation, a contradiction.

\subsection{Proof of Lemma~\ref{lem:t_inc}}

By Lemma~\ref{lem:t_fixed}~(b), $t_{\ell+1}$ is larger than
$t_{\ell}+\delta_{\min}^C$, provided no input transition occurs earlier.
By condition (S2) in the signal definition,
$t_{\ell}-t_{\ell-1}>0$
is guaranteed also in the latter case.

\subsection{Proof of Lemma~\ref{lem:fixed_is_fixed}}

Assume by contradiction that some iteration~$\ell \ge 2$ is the first
     in which a fixed transition is canceled; obviously, this can only
happen in step~(iii).
Thus, there exists a transition at time~$t_\ell$ that generated a new transition
     at some time~$t$ that results in the cancellation of a fixed transition at
     time~$t'$, i.e., $t\leq t'$. 
     Lemmas~\ref{lem:t_fixed} and~\ref{lem:minimal:delay} imply that
     $t_{\ell}-t' \leq -\delta_{\min} <0$
     in this case.
However, since both~$t_\ell$ and $t'$ are fixed and $t_\ell$ is the time for the
current iteration, by Lemma~\ref{lem:t_inc} $t'\leq t_\ell$
which provides the required contradiction.

\subsection{Proof of Theorem~\ref{thm:execution}}
 From Lemma~\ref{lem:t_inc}, we deduce that for all times~$T\ge 0$,
      there is an iteration~$\ell \ge 1$ such that~$t_\ell > T$ or the
      algorithm terminates.
      From Lemmas~\ref{lem:t_fixed}, \ref{lem:t_inc},
      and~\ref{lem:fixed_is_fixed}, we know that the algorithm does not add
transitions with times less than~$t_\ell$.

To prove
uniqueness of the execution,
assume by contradiction that there is a second execution and consider the
first differing transition.
By an (admittedly tedious) induction argument following the algorithm's steps,
one finds a contradiction to either (E3) or (E4).

\subsection{Proof of Theorem~\ref{thm:or:loop}}

If $\Delta$ is too small, then
the pulse is filtered by the channel in the feed-back loop. 
If it is too large,
the pulse is captured by the storage loop, leading to a stable output~1:
\begin{lemma}
\label{lem:big:pulse}
If the input pulse's length $\Delta$ satisfies $\Delta\geq\delta_\infty$, 
then the output of the OR has a unique rising transition at
time~0. 
\end{lemma}
\begin{proof}
%
%
%
%
Assigning the channel output~$s_c$ a single rising transition at time~$\delta_\infty$ is
part of a consistent execution, in which the OR's output has a single rising
transition at time~$0$.
The lemma now follows from uniqueness of executions.
\end{proof}

\begin{lemma}
\label{lem:small:pulse}
If the input pulse's length $\Delta$ satisfies
$\Delta \leq \delta_\infty - \delta_{\min}$, then the OR output 
contains only the input pulse.
\end{lemma}
\begin{proof}
The input signal contains only two transitions: One at time $t_1=0$ and one at
time $t_2=\Delta\leq \delta_\infty - \delta_{\min}$.
Since $\delta_1 = \delta_\infty$ and hence 
$t_2 = t_1 + \Delta \leq t_1 + \delta_1 - \delta_{\min}$,
the two pending transitions of $c$'s output cancel by Lemma~\ref{lem:minimal:delay},
and no further transitions are generated afterwards.
\end{proof}

The following lemma shows that the procedure for defining the $\Delta_n$ indeed stops if and only if
$\Delta_1\neq\tilde{\Delta}_1$, and can be used to bound the number of steps
until it stops.

\begin{lemma}\label{lem:exponential:decay}
$\lvert f(\Delta_1) - \tilde{\Delta}_1 \rvert \geq (1+\delta'(0)) \cdot \lvert
\Delta_1 - \tilde{\Delta}_1 \rvert$
if $\Delta_1>0$.
\end{lemma}
\begin{proof}
We have
\begin{equation}
\begin{split}
f'(\Delta_1)  = &
\big( 1 + \delta'(-\Delta_1) \big)\cdot \delta'\big(\Delta_1 -
\delta(-\Delta_1)\big)
\\ & + 1 +
\delta'(-\Delta_1)
\geq 1 + \delta'(0)
\end{split}
\end{equation}
because $\delta'(-\Delta_1)\geq \delta'(0)$  and $\delta'(T)>0$
for all~$T$ as $\delta$ is concave and increasing.
The mean value theorem of calculus now implies the lemma.
\end{proof}

If $\Delta_0\geq \delta_\infty$ or $\Delta_0\leq \delta_\infty -
\delta_{\min}$, then Lemmas~\ref{lem:big:pulse} and~\ref{lem:small:pulse}
show the theorem.

So let $\Delta_0\in (\delta_\infty - \delta_{\min} , \delta_\infty)$.
By Lemma~\ref{lem:exponential:decay}, the number of generated pulses until the
procedure stops is in the order of $\log 1 / \lvert \Delta_1 -
\tilde{\Delta}_1\rvert$.
Setting $g(\Delta_0) = \Delta_0 - \delta_\infty + \delta(\Delta_0 -
\delta_\infty)$ such that $\Delta_1=g(\Delta_0)$, 
and applying the mean value theorem of calculus to this
function, we see analogously as in the proof of
Lemma~\ref{lem:exponential:decay} that
\[
\lvert \Delta_1 - \tilde{\Delta}_1 \rvert
\geq 
\big( 1 + \delta'(0) \big) \cdot
\lvert \Delta_0 - \tilde{\Delta}_0 \rvert
\enspace.
\]
Hence the number of generated pulses is in the order of $\log 1 / \lvert
\Delta_0 - \tilde{\Delta}_0\rvert$.
Since both the length $\Delta_n$ of the occurring pulses 
and, by symmetry, the time between them is at most $\delta(0)$
since it would be captured otherwise, 
we have the same asymptotic bound on the stabilization
time.

\subsection{Proof of Lemma~\ref{lem:ht:exists}}

\begin{lemma}\label{lem:duty:cycle:delta}
Let~$c$ be an exp-channel~$c$ with threshold~$V_{th}$.
Then there exists some~$\Delta>0$ such that every periodic pulse train with pulse
lengths at most~$\Delta$ and duty cycle (ratio of 1-to-0) at most~$V_{th}$ is mapped to the zero
signal by~$c$.
\end{lemma}
\begin{proof}
Let $t_1,t_2,\dots$ be the times of transitions in the input pulse train with
duty cycle $\gamma\leq V_{th}$, i.e.,
$t_1=0$, $t_{2n+2} = t_{2n+1}+\Delta$, and $t_{2n+1} = t_{2n} + \Delta/\gamma$.
We assume that $\Delta$ is smaller than both $-\tau\log(1-V_{th})$ and a
to-be-determined $\Delta_0$, and show inductively that all pulses get canceled:

If $\Delta\leq -\tau\log(1-V_{th})$, then 
$\delta_1=\delta_\uparrow(\infty)=T_p-\tau\log(1-V_{th})>T_p$ and 
$\delta_2=\delta_\downarrow(\Delta-\delta_1) \leq T_p$, so the first pulse is canceled.
For the induction step, we assume $\delta_{2n}\leq T_p$ and find
\[
\begin{split}
\delta_{2n+1} & = \delta_\uparrow(\Delta/\gamma - \delta_{2n})
\geq \delta_\uparrow(\Delta/V_{th} - T_p)
\\ &
= T_p - \tau\log(1-V_{th}) + \tau\log(1-V_{th} e^{-\Delta/V_{th}\tau})
\end{split}
\]
Hence, $t_{2n+1}$ and $t_{2n+2}=t_{2n+1}+\Delta$ cancel if 
\[
\begin{split}
\Delta & \leq
\delta_{2n+1}-T_p
\\ &
=
-\tau\log(1-V_{th}) + \tau\log(1-V_{th} e^{-\Delta/V_{th}\tau})
\enspace,
\end{split}
\]
which is equivalent to 
\[
\begin{split}
h(\Delta)
=
V_{th} e^{-\Delta/V_{th}\tau}
+
(1-V_{th})e^{\Delta/\tau} \leq 1
\enspace.
\end{split}
\]
The latter satisfies $h(0)=1$ and 
\[
\begin{split}
h'(\Delta)
=
-\frac{1}{\tau} e^{-\Delta/V_{th}\tau}
+
\frac{1-V_{th}}{\tau}e^{\Delta/\tau}
\enspace,
\end{split}
\]
in particular $h'(0)=-V_{th}/\tau<0$.
There hence exists some~$\Delta_0>0$ such that $h(\Delta)\leq 1$ for all
$0\leq\Delta\leq\Delta_0$.
Thus, $t_{2n+1}$ and $t_{2n+2}$ cancel, and
\[
\begin{split}
\delta_{2n+2} & =
\delta(\Delta - \delta_{2n+1})
\leq
\delta(-T_p)=T_p
\end{split}
\]
because $h(\Delta)\leq1$, which completes the induction step.
\end{proof}

To prove Lemma~\ref{lem:ht:exists}, we use the notation of the proof 
of Lemma~\ref{lem:duty:cycle:delta}.
The unique root of $h'(\Delta)$ is 
\[
\Delta_\tau=-\frac{\tau\log(1-V_{th})}{1+1/V_{th}}
\enspace,
\]
which goes to infinity as $\tau\to\infty$.
We can choose $\Delta_0=\Delta_\tau$ because $h'(\Delta)\leq0$ for all
$0\leq\Delta\leq\Delta_\tau$.
Because also $-\tau\log(1-V_{th})$ goes to infinity as $\tau\to\infty$, we can
find, for any given $\hat{\Delta}$, some $\tau>0$ such that both $\hat{\Delta}\leq
-\tau\log(1-V_{th})$ and $\hat{\Delta}\leq \Delta_\tau$.
But for this $\tau$, all input pulse trains with pulse lengths at most~$\hat{\Delta}$ and
duty cycle at most $V_{th}=\gamma$ get mapped to the zero signal, as asserted.

\subsection{Proof of Theorem~\ref{thm:channels:are:continuous}}
\label{sec:contproof}
For ease of exposition and for space reasons, we give the proof only in the
case of symmetric channels, i.e., for the case that
$\delta_\uparrow=\delta_\downarrow=\delta$.

We begin by noting that channels are monotone.
To compare certain signals, we write $s_1\leq s_2$ if~$s_2$ is~$1$
whenever~$s_1$ is.

\begin{lemma}
\label{lem:channels:are:increasing}
Let~$s_1$ and~$s_2$ be signals such that $s_1\leq s_2$ and let~$c$ be a
channel.
Then, $c(s_1)\leq c(s_2)$.
\end{lemma}

The following Lemma~\ref{lem:add:pulse:to:the:end}
provides an optimal choice for appending a pulse at the
end of a signal in order to maximize $\mu_T$.
We will use it later when bounding the maximum impact of an infinitesimally small
pulse.

We use the shorthand notation $(x)_+$ for $\max(x,0)$.

\begin{lemma}
\label{lem:add:pulse:to:the:end}
Let~$s$ be a signal that is eventually constant~$0$ and let~$c$ be a channel.
Denote by~$t_n$ the time of the last (falling) transition in~$s$ and
by~$\delta_n$ its
delay in the channel algorithm for~$c$. 
Then, the maximal $\mu_T(c(s'))$ among
all~$s'$ obtained from~$s$ by appending one pulse of length~$\Delta$ after
time~$t_n$ is attained by the addition of the pulse at time $t_n +(\delta_n -
\delta_{\min})_+$ (which results in a cancellation, i.e., a right-shift, of the 
last transition).
\end{lemma}
\begin{proof}
We first show the lemma for $T=\infty$ and then extend the result to
finite~$T$.
Let~$s'_\gamma$ be the addition of the pulse of length~$\Delta$ to~$s$ at time
$t_n + \gamma$.

For all $0\leq\gamma\leq\delta_n-\delta_{\min}$, set
\[
f(\gamma) = \gamma + \delta(\Delta - \delta(\gamma - \delta_n))
\enspace.
\]
In the class of all~$s'_\gamma$ with $0\leq\gamma\leq \delta_n-\delta_{\min}$
(which can be empty), the
maximum of~$\mu_\infty(c(s'_\gamma))$ is attained at the maximum of~$f$.
This is because the transition at time $t_n+\gamma$ cancels that at time~$t_n$
in this case.
The derivative of~$f$ is equal to
\[
f'(\gamma) = 1 - \delta'(\Delta - \delta(\gamma - \delta_n))\cdot
\delta'(\gamma - \delta_n)
\enspace.
\]
The condition $f'(\gamma) = 0$ is equivalent to $\delta'(\Delta - \delta(\gamma
- \delta_n)) = 1 / \delta'(\gamma - \delta_n)$, which is in turn equivalent to
  $\delta(\Delta - \delta(\gamma - \delta_n)) = -(\gamma - \delta_n)$, i.e.,
$\Delta=0$, as $\delta'(t)=1/\delta'(-\delta(t))$ by Lemma~\ref{lem:delta:min}.
Hence, $f'(\gamma)$ is never zero.
Since $f'(\gamma)\to 1$ as $\gamma\to\infty$, as the concave $\delta$ satisfies
$\lim_{t\to\infty}\delta'(t)=0$, the derivative of~$f$ is always
positive, hence~$f$ is increasing.
This shows that $\gamma = \delta_n-\delta_{\min}$ is a strictly better choice
than any other $\gamma$ in this class.

For the class of~$s'_\gamma$ with $\gamma \geq (\delta_n - \delta_{\min})_+>0$, we
define the function
\[
g(\gamma) = \Delta+ \delta(\Delta - \delta(\gamma - \delta_n)) -
\delta(\gamma - \delta_n)
\enspace.
\]
Since the transitions at~$t_n$ and $t_n+\gamma$ do not cancel in this class,
the maximum of $\mu_\infty(c(s'_\gamma))$ is attained at the maximum of~$g$.
But it is easy to see, using the monotonicity of~$\delta$, that~$g$ is decreasing.
The maximum of~$g$ is hence attained at $\gamma = (\delta_n- \delta_{\min})_+$. 

Consequently, the choice $\gamma = \gamma_0 = (\delta_n- \delta_{\min})_+$
maximizes $\mu_\infty(c(s'_\gamma))$ in any case. By Lemma~\ref{lem:minimal:delay}, this choice
results in a cancellation of the last (falling) transition in $s$, hence
a right-shift of the latter in $s'$. This concludes our proof for $T=\infty$.

Let now~$T$ be finite.
Denote by~$T_0$ the time of the last, falling, output transition in
$c(s_{\gamma_0}')$.
In this case, transitions of $c(s)$ and $c(s_{\gamma_0}')$ are the same except
the last,
falling, transition which is delayed from $t_n + \delta_n$ to~$T_0$.
We distinguish the two cases (a) $T\leq T_0$ and (b) $T > T_0$.
In case~(a), the last transition of $c(s)$ is delayed beyond~$T$ in
$c(s_{\gamma_0}')$.
Because all other transitions remain unchanged in all $c(s_\gamma')$, the
measure $\mu_T(c(s_{\gamma_0}'))$ is maximal among all $\mu_T(c(s_\gamma'))$
if $T\leq T_0$.
In case~(b), we have $\mu_T(c(s_{\gamma_0}')) = \mu_\infty(c(s_{\gamma_0}'))$.
But because $\mu_T \leq \mu_\infty$ and $\mu_\infty(c(s_{\gamma_0}'))$ is
maximal among all $\mu_\infty(c(s_{\gamma}'))$, so is
$\mu_T(c(s_{\gamma_0}'))$ among all $\mu_T(c(s_{\gamma}'))$.
\end{proof}

We next effectively bound the maximum impact on~$\mu_T$ that a set of 
pulses of small combined length can have.

\begin{lemma}
\label{lem:epsilon:pulses:at:the:end}
Let~$s$ be a signal that is eventually constant~$0$ and let~$c$ be a channel.
Then there exists a constant~$d$ such that
the maximal~$\mu_T(c(s'))$ among all~$s'$ obtained from~$s$ by adding
pulses of combined length~$\varepsilon$ after the last transition of~$s$ is at
most $\mu_T(c(s)) + d\cdot\varepsilon$.
\end{lemma}
\begin{proof}
It suffices to show the lemma for $T=\infty$.
Let~$\varepsilon=\sum_{k=1}^K \varepsilon_k$ be any decomposition
of~$\varepsilon$ into positive pulse length.
We add, one after the other, pulses of length~$\varepsilon_k$ after the last
transition.
We show that the maximum gain in measure after adding~$L$ pulses is at most
$d\cdot\sum_{k=1}^L \varepsilon_k$ for a constant~$d$ depending on~$c$
and~$s$.

Denote by~$t_n$ the last transition in~$s$ and by~$\delta_n$ its delay.
By Lemma~\ref{lem:add:pulse:to:the:end}, it is optimal to add the first pulse
(of length~$\varepsilon_1$) at time $t_n + (\delta_n-\delta_{\min})_+$;
call the resulting signal~$s_1'$.

We first assume $\delta_n-\delta_{\min}\geq0$.
Here, the two new transitions in~$s_1'$ are $t_{n+1} = t_n + \delta_n-\delta_{\min}$ and
$t_{n+2} = t_n + \delta_n-\delta_{\min} + \varepsilon_1$.
Their corresponding delays are $\delta_{n+1}=\delta_{\min}$ and $\delta_{n+2} =
\delta(\varepsilon_1 - \delta_{\min})$.
By the mean value theorem of calculus and Lemma~\ref{lem:delta:min}, the
duration of the resulting pulse is
\[
\begin{split}
\delta_{n+2} - \delta_{n+1}  = 
\delta(\varepsilon_1 -\delta_{\min}) - \delta(- \delta_{\min})
 = \varepsilon_1\cdot \delta'(\xi)
\end{split}
\]
for some $-\delta_{\min} \leq \xi \leq \varepsilon_1 - \delta_{\min}$.
Since~$\delta'$ is decreasing and $\delta'(-\delta_{\min})=1$
according to Lemma~\ref{lem:delta:min}, we hence deduce
$0\leq\delta_{n+2} - \delta_{n+1} \leq \varepsilon_1$.
Thus $\mu_T(c(s_1') - c(s)) = \varepsilon_1 + \delta_{n+2} - \delta_{n+1} \leq
2\varepsilon_1$.
Since $\delta_{n+2} > \delta_{\min}$, we can 
continue this argument inductively.

If now $\delta_n - \delta_{\min}<0$, then~$t_n$ is effectively replaced  by
$t_n+\varepsilon_1$ in~$s_1'$, i.e., right-shifted.
This changes the measure by
\[
\big(
\delta(t_n - t_{n-1} + \varepsilon_1 - \delta_{n-1})
-
\delta(t_n - t_{n-1} - \delta_{n-1})
\big)_+
\enspace,
\]
which is at most $\varepsilon_1 \cdot \delta'(t_n - t_{n-1} - \delta_{n-1})$ by
the mean value theorem.
We note that this second case only occurs until the first case
happens one time.
We can hence merge all the~$\varepsilon_k$ of the first case and set $d =
\max(2, \delta'(t_n - t_{n-1} - \delta_{n-1}))$.
\end{proof}

We combine the previous two lemmas to show continuity of channels:

\begin{proof}[of Theorem~\ref{thm:channels:are:continuous}]
Let~$s$ be a signal.
We show that, if $\lVert s - s_n \rVert_T \to 0$, then $\lVert c(s) - c(s_n)
\rVert_T \to 0$.
Because
\[
\lvert s - s_n\rvert = (\max(s,s_n) - s) + (s - \min(s,s_n))
\enspace,
\]
where $\max(s,s_n)(t)=\max(s(t),s_n(t))$ and $\min(s,s_n)(t)=\min(s(t),s_n(t))$ for all $t$,
the condition $\lVert s-s_n\rVert_T\to 0$ is equivalent to the conjunction of
both $\lVert s
- \max(s,s_n)\rVert_T\to0$ and $\lVert s - \min(s,s_n)\rVert_T\to0$.
Because
$\max(c(s),c(s_n)) \leq c(\max(s,s_n))$
and
$\min(c(s),c(s_n)) \geq c(\min(s,s_n))$
by Lemma~\ref{lem:channels:are:increasing},
we have
\[
\begin{split}
\lvert c(s) - c(s_n) \rvert
\leq
\ &
c(\max(s,s_n)) - c(s)
\\&+ c(s) - c(\min(s,s_n))
\enspace,
\end{split}
\]
which shows that
we can
suppose without loss of generality $s_n\geq s$ for all~$n$.

Let $(t_m,0),(t_{m+1},1)$ be a negative pulse in~$s$.
Since there are only finitely many negative pulses before time~$T$ because~$T$
is finite by assumption, it suffices
to show $\mu_T(c(s_n)-c(s))\to0$ in the case that $s_n-s$ is zero
outside of $[t_m,t_{m+1}]$, i.e., that the only additions of~$s_n$ with respect
to~$s$ lie in the given negative pulse.

Let~$\mu_T(s_n-s) \leq \varepsilon$.
It follows from Lemma~\ref{lem:epsilon:pulses:at:the:end} that the increase in
measure incurred directly from the new pulses is $O(\varepsilon)$.
Furthermore, by Lemma~\ref{lem:add:pulse:to:the:end}, the measure incurred by
later transitions~$t_k$ with~$k>m$
are biggest when merging all new pulses at the end of the negative pulse.
Because the delays of these transitions
depend continuously on~$\varepsilon$ and~$\mu_T(c(s_n)-c(s))$ depends
continuously on these delays, we have $\mu_T(c(s_n)-c(s))\to 0$ as
$\varepsilon\to0$.
\end{proof}

\subsection{Proof of Theorem~\ref{thm:no_forward_circuit}}
Suppose that there exists a forward circuit that solves bounded SPF with
stabilization time bound~$K$.
Denote by~$s_\Delta$ its output signal when feeding it a $\Delta$-pulse at
time~$0$ as the input.
Because~$s_\Delta$ in forward circuits is a finite composition of continuous
functions by
Theorem~\ref{thm:channels:are:continuous}, the measure $\mu_T(s_\Delta)$
depends continuously on~$\Delta$.

By the nontriviality condition (F3) of the SPF problem, there exists
some~$\Delta_0$ such that $s_{\Delta_0}$ is not zero.
Set $T = 2\Delta_0 + K$.

Let~$\varepsilon>0$ be smaller than both~$\Delta_0$ and $\mu_T(s_{\Delta_0})$.
We show a contradiction by finding a~$\Delta$ such that~$s_\Delta$ either
contains a pulse of length less than~$\varepsilon$ (contradiction to the no
short pulses condition (F4)) or contains a transition
after time $\Delta+K$ (contradicting the bounded stabilization time
condition~(F5)).

Since $\mu_T(s_\Delta)\to0$ as $\Delta\to0$ by the no generation condition (F2)
of SPF,
there exists a~$\Delta_1<\Delta_0$ such that $\mu_T(s_{\Delta_1})=\varepsilon$
by the intermediate value property of continuity.
By the bounded stabilization time condition (F5), there are no transitions
in~$s_{\Delta_1}$ after time $\Delta_1+K$.
Hence, $s_{\Delta_1}$ is~$0$ after this time because otherwise it is~$1$ for the
remaining duration $T - (\Delta_1+K) > \Delta_0 > \varepsilon$, which would
mean that $\mu_T(s_{\Delta_1})>\varepsilon$.
Consequently, there exists a pulse in~$s_{\Delta_1}$ before
time $\Delta_1+K$.
But any such pulse is of length at most~$\varepsilon$ because
$\mu_{\Delta_1+K}(s_{\Delta_1}) \leq \mu_T(s_{\Delta_1})=\varepsilon$.
This is a contradiction to the no short pulses condition (F4).

\subsection{Proof of Theorem~\ref{thm:simulation}}
We will show the statement by induction on~$d(e) \ge 0$ for the case
     where~$e$ is a transition in~$C$'s execution.
The proof for the case where~$e$ is a transition in~$C_k(o)$'s
     execution is analogous.

{\em Induction base:\/} By the construction of an unrolled circuit, $s_{\sigma}$
     and~$s_{\sigma'}$ have the same initial transitions if~$\sigma$
     is a gate or channel, and the same transitions
     if~$\sigma=\sigma'$ is the input.
Since, for all these transitions~$e$, $d(e) = 0 \le z(\sigma')$, the
     statement holds for~$d(e) = 0$.

{\em Induction step:\/} Assume that the theorem holds for all
     transitions~$e$ with depth~$d(e) \le k$.
We show that it also holds for transitions~$e'$ with $d(e')=k+1$
     if~$k+1\le z(\sigma')$.
If~$e'$ is a transition initially added by the constructive algorithm
     at time~$0$, $d(e')=1$.
From the definition of the unrolling, we immediately obtain that~$e'$
     is also added to~$s_{\sigma'}$ if~$z(\sigma') \ge 1$.
Otherwise, $e'$ must have been added within an iteration of the
     constructive algorithm.
Assume by contradiction that~$e'$ is the first transition (in the
     order transitions are generated by the constructive algorithm)
     with causal depth~$k+1$ added to a signal in~$C$ but not added
     to the respective signal in~$C_k(o)$.
We distinguish two cases for~$\sigma$: 

If~$\sigma=(v,w)$ is a channel in~$C$: Transition~$e'$ may only have
     been added to~$s_{(v,w)}$ by the channel algorithm with input
     signal~$s_v$ and a current transition~$e''$ with $d(e'')
     = k$.
The time of transition~$e'$ depends on the time of~$e''$ and on the
     last output transition only.
From the fact that~$e'$ is the first with depth~$k+1$ not added
     to~$s_{\sigma'}$, and the induction hypothesis applied to
     both~$s_{v}$ and~$s_{(v,w)}$, we deduce that transition~$e'$ is
     also added to~$s_{\sigma'}$; a contradiction to the assumption
     that it is the first one not added.

Since~$\sigma$ cannot be the input port, because all its transitions
     have causal depth~$0$ only, the case of~$\sigma$ being a gate in~$C$
     remains: However, then~$e'$ is generated due to a transition~$e''$ on
     a predecessor~$w$ of~$\sigma$ in~$C$.
Further, $d(e'') \le k+1$ and~$z(w') \ge k+1$ must hold for
     vertex~$w'$ corresponding to~$w$.
Since, $e'$ is the first transition with causal depth less or
     equal~$k+1$ not added to both circuits' signals, $e''$ was added to
     both~$s_{w}$ and~$s_{w'}$, and thus~$e'$ is added to~$s_{\sigma'}$; 
     a contradiction that it is the first one not added.

This completes the induction step.

\subsection{Proof of Theorem~\ref{thm:main_impossibility}}

Let the {\em aligned bounded SPF problem\/} be the SPF problem with
     the following modifications: We first require that if the input
     signal is a pulse, then the pulse must start at time~$0$ and be
     of length at most~$1$.
We call such signals {\em valid\/} input signals.
Further, we require that if the output signal~$o$ makes a transition
     to~$1$, it must do so before time~$K+1$, and~$o$ must remain~$1$
     from thereon until time~$K+2$, from whereon it is~$0$ until
     time~$K+3$ followed by a pulse of length~$1$ at time~$K+3$.
If the input is constant~$0$, we require that the output is a pulse of
     length~$1$ at time~$K+3$.
From every circuit that solves the (original) bounded SPF problem, 
it is not difficult
     to build a circuit that solves the aligned SPF
     problem by adding adequate capturing and suppression circuitry. 
In the following, we show that no circuit solves the aligned version of
     bounded SPF and thus, by the above reduction, the original
     bounded SPF problem.

Let~$C$ be a circuit that solves the aligned bounded SPF problem.
Then, for all input signals, the output signal~$o$ of~$C$ always
     contains a transition~$(K+4,0)$, regardless of the input.
Let~$D^C_o(i)$ be the causal depth of this transition in circuit~$C$
     when the input signal is~$i$.

\begin{lemma}\label{lem:clockbound}
Let~$C$ be a circuit that solves the aligned bounded SPF problem.
Then there exists an input signal~$i$ such that $D^C_o(i) >
     (K+5)/\delta_{\min}+2$.
\end{lemma}
\begin{proof}
Let $N = (K+5)/\delta_{\min}+2$, and assume for a contradiction that
$D^C_o(i) \le N$ for all valid input signals~$i$.
Consider the $N$-unrolled circuit~$C_N(o)$.
From Theorem~\ref{thm:simulation}, we obtain that transition~$(K+4,0)$,
     with causal depth in~$C$ at most~$N$ occurs at output~$o$ of~$C$
     if and only if it occurs at~$C_N(o)$'s output~$o'$ corresponding
     to~$o$.
Recalling Lemma~\ref{lem:Depth_Iteration}.(b), we obtain that the same holds
     for all transitions at output~$o$ with times less than~$K+4$;
     i.e., $C$'s and $C_N(o)$'s output signals restricted to
     time~$[-\infty,K+4]$ are the same for all valid input
     signals~$i$.
One can easily extend the forward circuit~$C_N(o)$ such that it remains a
     forward circuit and solves aligned bounded SPF, by suppressing all
     transitions at the output that possibly occur after time~$K+4$.
Since Theorem~\ref{thm:no_forward_circuit} also holds for the aligned
     bounded SPF problem, no such forward circuit exists; a contradiction
     to the initial assumption.
\end{proof}

From Lemma~\ref{lem:clockbound} and~\ref{lem:Depth_Iteration}, we
     obtain that, for input~$i$, the constructive algorithm does not
     mark fixed the output transition~$(K+4,0)$ before 
     iteration~$(K+5)/\delta_{\min}+2$.
However, from Lemma~\ref{lem:t_inc} and the fact that input signal~$i$
     contains at most~$2$ transitions besides the initial transition
     at~$-\infty$, we conclude that all iterations~$\ell \ge
     (K+5)/\delta_{\min}+2$ have~$t_\ell \ge K+5$.
From Lemma~\ref{lem:t_fixed}, we conclude that all transitions still
     existent at the end of these iterations must have times at
     least~$K+5$; a contradiction to the fact that the final output transition
     occurs at time~$K+4$.

This completes the proof of Theorem~\ref{thm:main_impossibility}.

}
{}

\end{document}